\documentclass[9pt,preprint]{sigplanconf} % prelude <<<
\usepackage{amsmath}
\usepackage{amssymb}
\usepackage{amsthm}
\usepackage{balance}
\usepackage{booktabs}
\usepackage{microtype}
\usepackage{pifont}
\usepackage{rgalg}
\usepackage{stmaryrd}
\usepackage{thm-restate}
\usepackage{thmtools}
\usepackage{tikz}
\usepackage{xcolor}

\usetikzlibrary{arrows}

\usepackage{hyperref}

\definecolor{darkblue}{rgb}{0,0,0.5}
\hypersetup{colorlinks,linkcolor=darkblue,citecolor=darkblue,urlcolor=darkblue}
\DeclareMathOperator*\argmax{arg\,max}

\DeclareMathOperator\E{E}

\declaretheorem{theorem}

\declaretheorem[sibling=theorem]{corollary}
\declaretheorem[sibling=theorem,style=definition]{definition}

\declaretheorem[sibling=theorem]{lemma}
\declaretheorem[sibling=theorem,style=definition]{problem}
\declaretheorem[sibling=theorem]{proposition}
\declaretheorem[sibling=theorem,style=remark]{remark}

\newcommand\draft[1]{}

\newcommand\longshort[2]{#1} % #1 = long, #2 = short

\newcommand*\defeq{\mathrel{:=}}
\newcommand*\df[1]{\emph{\textbf{#1}}}

\newcommand*\Iff{\Leftrightarrow}
\newcommand*\Implies{\Rightarrow}
\newcommand*\maxsat{\textsc{MaxSAT}}
\newcommand*\randomG{{\bf G}}
\newcommand*\randomH{{\bf H}}
\newcommand*\reach[2]{{\cal R}_{#1}#2}
\newcommand*\R{\mathbb{R}}

\renewcommand\iff{\leftrightarrow}
\renewcommand\implies{\to}

\newcommand\bigcomment[2]{\draft{%
  \smallskip%
  \par\vbox{\hrule height1pt\hbox to\hsize{\advance\hsize by-10pt%
    \vrule width1.4pt%
    \hfil\vbox{\smallskip\noindent\textcolor{red}{\bf#1:} #2\smallskip}\hfil%
    \vrule width1.4pt%
  }\hrule depth1pt}%
  \smallskip}%
} %RG: you can say "\RG{Foo? \HY{Bar!}}" and it looks like a thread in a forum

\newcommand{\RG}[1]{\bigcomment{RG}{#1}}

\title{Abstraction Refinement Guided by a Learnt Probabilistic Model}
\authorinfo{Radu Grigore \and Hongseok Yang}{Oxford University}

\draft{\overfullrule=10pt}

\begin{document}
\toappear{}
\maketitle

\begin{abstract}
The core challenge in designing an effective static program
analysis is to find a good program abstraction -- one that
retains only details relevant to a given query.
In this paper,
  we present a new approach for automatically finding such an abstraction.
Our approach uses a pessimistic strategy,
  which can optionally use guidance from a probabilistic model.
%The probabilistic model is tuned by observing previous runs of the analysis.
Our approach applies to parametric static analyses implemented in Datalog,
and is based on counterexample-guided abstraction refinement.
For each untried abstraction, our probabilistic model provides a probability of success,
while the size of the abstraction provides an estimate of its cost in terms of analysis time.
Combining these two metrics, probability and cost,
our refinement algorithm picks an optimal abstraction.
Our probabilistic model is a variant of the Erd\H{o}s--R\'enyi random graph model,
  and it is tunable by what we call hyperparameters.
We present a method to learn good values for these hyperparameters,
  by observing past runs of the analysis on an existing codebase.
We evaluate our approach on an object sensitive
  pointer analysis for Java programs, with two client analyses (PolySite and Downcast).

% vim:spell spelllang=en_gb:

\end{abstract}


\section{Introduction}\label{sec:introduction}
We wish that static program analyses would become better as
they see more code. Starting from this motivation, we designed
an abstraction refinement algorithm that incorporates knowledge
learnt from observing its own previous runs, on an existing
codebase. For a given query about a program, this
knowledge guides the algorithm towards a good abstraction that
retains only the details of the program relevant to the query.
Similar guidance also features in existing abstraction refinement
algorithms~\cite{slam:popl02,clarke:jacm03,blast:popl04},
but is based on nontrivial heuristics that are developed manually by
analysis designers.
These heuristics are often suboptimal and difficult to transfer from one analysis to another.
Our algorithm has the potential to improve itself by learning from past runs,
  and it applies to almost any analysis implemented in Datalog.

Prior work on abstraction refinement for Datalog~\cite{zhang:pldi14}
  implicitly uses an optimistic strategy:
  the search is geared towards finding an abstraction
    that would show the current counterexample to be spurious.
We take the complimentary approach:
  our search is geared towards finding an abstraction
    that would show the current counterexample to be unavoidable.
Furthermore, we bias the search by using a probabilistic model,
  which is tuned using information from previous runs of the analysis.

%In this paper we present our abstraction refinement algorithm
%and its companion probabilistic model. Our algorithm applies to
%any parametric static analysis implemented in Datalog,
%provided that its precision increases when the values of the
%parameters increase. Such analyses are typically run in a loop that
%iteratively refines the parameter setting. Our idea is to equip such
%an analysis with a probabilistic model
%that can predict, for every untried parameter setting, what would happen
%if the analysis were run with the setting.
%The model makes this prediction
%using information found by the analysis in the failed
%iterative process so far, and guides the analysis when it chooses
%a next parameter setting.
%The probabilistic model itself is parametrised.
%To distinguish the parameters of the analysis from those of the probabilistic model,
%  we call the latter hyperparameters.
%Good values for
%the hyperparameters are learnt by observing runs of the analysis
%on an existing codebase, and are later used when new unseen programs are analysed.

In other approaches to program analysis that are based on learning%
  ~\cite{yi:ipl07,Raychev-popl15},
  the analysis designer must choose appropriate features.
A feature is a measurable property of the program, usually a numeric one.
Choosing features that are effective for program analysis is nontrivial,
  and involves knowledge of both the analysis and the probabilistic model.
In our approach,
  the analysis designer does not need to choose appropriate features.

Instead of observing features,
  our models observe directly the internal representations of analysis runs.
Parametric static analyses implemented in Datalog
consist of universally quantified Horn clauses, and work by instantiating
the universal quantification of these clauses, while respecting the constraints on
instantiation imposed by a given parameter setting. These instantiated
Horn clauses are typically implications of the form
$$
h \leftarrow t_1,t_2,\ldots,t_n
$$
and can be understood as a directed (hyper) arc from the source
vertices $t_1,\ldots,t_n$ to the target vertex $h$.
Thus, the instantiated Horn clauses taken altogether form a hypergraph.
This hypergraph changes when we try the analysis again
with a different parameter setting. Given a hypergraph obtained under one
parameter setting, we build a probabilistic model that predicts how
the hypergraph would change if a new and more precise parameter setting were used.
In particular, the probabilistic model estimates how likely it is that the new
parameter setting will end the refinement process, which happens when the new
hypergraph includes evidence that the analysis will never prove a query.
Technically, our probabilistic model is a variant of the Erd\H{o}s--R\'enyi random graph model%
  ~\cite{erdos-renyi}:
given a template hypergraph~$G$, each of its subhypergraphs~$H$ is assigned a probability,
which depends on the values of the hyperparameters.
Intuitively, this probability quantifies the chance that
$H$~correctly describes the changes in~$G$ when the analysis is run with
the new and more precise parameter settings. The hyperparameters quantify how much approximation
occurs in each of the quantified Horn clauses of the analysis.
We provide an efficient method for learning hyperparameters
from prior analysis runs.
Our method uses certain analytic bounds in order to
avoid the combinatorial explosion of a naive learning method
based on maximum likelihood;
the explosion is caused by $H$~being a latent variable,
  which can be observed only indirectly.

The next parameter setting to try is chosen by our refinement algorithm
  based on predictions of the probabilistic model
  but also based on an estimate of the runtime cost.
For each parameter setting,
  the probability of successfully handling the query is evaluated by our model,
  and the runtime is estimated to increase with the precision of the parameter setting.
We prove that our method of integrating these two metrics is optimal,
  under reasonable assumptions.

The paper starts with an informal overview of our approach (\autoref{sec:overview})
and a review of notations from probability theory (\autoref{sec:preliminaries}),
and is followed by a description of our probabilistic model
(\autoref{sec:model})
and its learning algorithm (\autoref{sec:learning}).
The probabilistic model is then used to implement a refinement loop
  that optimally chooses the next parameter setting (\autoref{sec:refinement}).
The experimental evaluation (\autoref{section:experiments})
  shows the value of the pessimistic strategy,
  but suggests we need better optimisers
    in order to take full advantage of the probabilistic model.
\autoref{sec:related} positions our work in the various
attempts to combine probabilistic reasoning and static analyses, and
\autoref{sec:conclusions} concludes the paper.
Most proofs are in \longshort{appendices}{the full version}.

% vim:spell spelllang=en_gb:

\section{Overview}\label{sec:overview} % <<<

\begin{figure} % fig:architecture <<<
\begin{center}
\begin{tikzpicture}[align=center,>=stealth',yscale=0.9375]
\node (program) at (0,2) {program};
\node (query) at (0,-2) {query};
\node (OK) at (-2,2) {\textcolor{green}{\ding{51}}};
\node (NOK) at (5,2) {\textcolor{red}{\ding{55}}};
\node[draw] (datalog) at (0,0) {Static analyser\\(Datalog solver)};
\node[draw] (maxsat) at (3,0) {Optimiser\\ ($\maxsat$ solver)};
\node[fill=black!10,rounded corners] (model) at (3,-2) {Learnt probabilistic\\ model};
\draw[thick,->] (datalog) to[bend left=50] node[above]{counterexample} (maxsat);
\draw[thick,->] (maxsat) to[bend left=50] node[below]{abstraction} (datalog);
\draw[thick,->] (program) -- (datalog);
\draw[thick,->] (query) -- (datalog);
\draw[line width=3pt,<->,draw=black!50] (maxsat) -- (model);
\draw[thick,->,red] (maxsat) -- (NOK);
\draw[thick,->,green] (datalog) -- (OK);
\end{tikzpicture}
\end{center}
\caption{Architecture}
\label{fig:architecture}
\end{figure}
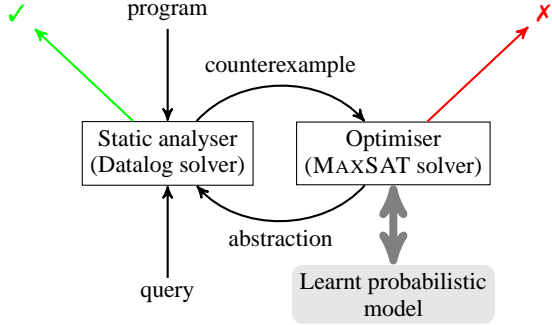
% >>>a

\autoref{fig:architecture} gives a high level overview of our abstraction
refinement algorithm, and in particular it shows the role of our probabilistic model.
The refinement loop
is standard, with analysis on one side and refinement on the other. Our contribution
lies in the refinement part, which receives guidance from a learnt probabilistic
model and chooses the next abstraction by balancing the model's prediction
and the estimated cost of running the analysis under each abstraction.

We assume that the analysis is given and obeys two constraints.  The
first is that the analysis is implemented in Datalog -- it is specified in terms of universally quantified
Horn clauses, such as
\begin{align}
\begin{aligned}
  {\tt pointsto}(\alpha, \ell)
    \gets\;
    &
    {\tt precise}(\alpha),
    {\tt pointsto}(\beta,\ell),
    \\
    &
    {\tt assignTo}(\beta,\alpha)
\end{aligned}
\label{eqn:rule-example0}
\end{align}
in which all the free variables $\alpha,\beta,\ell$ are implicitly universally quantified.
We call these clauses \df{Datalog rules}.
The analysis works by instantiating the quantification of these rules,
  and thus deriving new facts.
A query is a particular fact such as ${\tt pointsto}(x,h)$,
  which is an instantiation of the left side of the rule~\eqref{eqn:rule-example0},
  with $\alpha := x$ and $\ell := h$.
The query represents an undesirable situation in the program being analysed.
The analysis could derive the query because the undesirable situation really occurs at runtime.
But, the analysis could also derive the query because it approximates the runtime semantics.
Our task is to decide whether it is possible to avoid deriving the query by approximating less.
If the query is derived,
  then the set of all instances of Datalog rules constitute a counterexample,
  which is then used for refinement.

The second constraint is that the analysis is parametric. For instance, it might have a
parameter for each program variable, which specifies whether the variable should be
tracked precisely or not. The analysis would encode a setting of these parameters
in Datalog by using relations ${\tt cheap}$ and ${\tt precise}$.
In fact, the Datalog rule~\eqref{eqn:rule-example0}
assumes such parametrisation and fires only when the parameter
setting dictates the precise tracking of the variable~$\alpha$.
For a parametric analysis, an abstraction can be specified by a parameter setting,
  and so we use these two terms interchangeably.

The refinement part analyses a counterexample,
  and suggests a new promising parameter setting.
If the counterexample derives the query without relying on approximations,
  then the refinement part reports impossibility and stops%
  ~\cite{zhang:pldi13,zhang:pldi14,terauchi:sas15}.
If the counterexample derives the query by relying on approximations,
  then the refinement part sets itself the goal
    to find a similar counterexample that does not rely on approximations.
This is a pessimistic goal.
To find such a similar counterexample,
  the analysis must be run with a different parameter setting.
Which one?
On the one hand, the parameter setting should be likely to uncover a similar counterexample.
On the other hand, the parameter should be as cheap as possible.
The refinement part uses a $\maxsat$ solver to balance these desiderata.
% RG: The above is supposed to be the same as below, but emphasizing 'pessimistic'.
%The refinement part analyses a counterexample,
%  and suggests a new promising parameter setting.
%When it receives a counterexample from the
%analysis, it first checks whether
%using a more precise parameter settings could remove this counterexample. If not,
%the refinement part reports
%impossibility and stops~\cite{zhang:pldi13,zhang:pldi14,terauchi:sas15}.
%Otherwise, it moves on to its main task of choosing a next parameter setting.
%The refinement part consults our probabilistic model, which uses the counterexample
%and predicts, for each untried parameter setting, the probability that
%the analysis under that setting ends the refinement loop. Based on the outcome of this
%consultation, the refinement part formulates an optimisation goal that balances
%these probabilities and the estimated costs of running the analysis under
%untried parameter settings. The resulting optimisation problem is then solved by a weighted $\maxsat$
%solver, and its solution becomes the next parameter setting that the analysis tries.

\begin{figure} % fig:running-example <<<
\small
\begin{verbatim}
    object x, y, z, v
    assume x.dirty
    x.value := 10
0:  smudge2(x, y)
0': y.value := y.value + 2 * x.value
1:  smudge3(y, z)
    if z.dirty && y.value > 5
      v.value := x.value + y.value
2:  smudge3(z, v)
    ...
3:  smudge5(x, y)
    ...
4:  smudge7(y, v)
    assert !v.dirty
\end{verbatim}
\caption{Example program to analyse}
\label{fig:running-example}
\end{figure}
% >>>

Consider now the example program in \autoref{fig:running-example}.
The language is idiosyncratic, and so will be the analysis.
The language and the analysis are chosen to allow a concise rendering of the main ideas.
In this toy language, each object has two fields,
  the boolean $\it dirty$ and the integer $\it value$.
Initially, all $\it value$ fields are~$0$.
Object $x$ is dirty at the beginning,
  and we are interested in whether object $v$ is dirty at the end.
Dirtiness is propagated from one object to another
  only by the primitive commands $\tt smudgeK$.
The effect of the command ${\tt smudgeK}(x,y)$
  is equivalent to the following pseudocode:
\begin{align*}
  &\text{if $(x.{\it value}+y.{\it value})\bmod K=0$} \\
  &\qquad\text{$y.{\it dirty} := x.{\it dirty} \lor y.{\it dirty}$}
%  &\qquad\text{$d := (x.{\it dirty}\lor y.{\it dirty})$};\; \text{$x.{\it dirty} := d;\; y.{\it dirty} = d$}
\end{align*}
That is,
  if the sum of the values of objects $x$ and $y$ is a multiple of~$K$,
  then dirt propagates from~$x$ to~$y$.

To decide whether object $v$ is dirty at the end,
  an analysis may need to track the values of multiple objects.
The values can be changed by guarded assignments.
The guard of an assignment can be any boolean expression;
  the right hand side of an assignment can be any integer expression.
In short, tracking values and relations between values could be expensive.

\begin{figure} % fig:running-example-picture <<<
\begin{center}
\begin{tikzpicture}[thick,>=stealth',shorten <=1pt,shorten >=1pt,auto]
  \draw[line width=2mm, black!20] (0,5) -- (1,4) -- (1,1) -- (3,0);
  \foreach \x/\n in {0/x, 1/y, 2/z, 3/v} {
    \node[anchor=south,text height=1.5ex,text depth=.5ex] at (\x,5.5) {$\n$};
  }
  \foreach \x in {0,1,2,3} {
    \draw[dashed] (\x,0) -- (\x,5);
  }
  \foreach \x/\xx/\y in {0/1/4, 1/2/3, 2/3/2, 0/1/1, 1/3/0} {
    \draw[shorten <=3pt,shorten >=3pt] (\x,\y+1) -- (\xx,\y);
%    \draw[shorten <=3pt,shorten >=3pt] (\xx,\y+1) -- (\x,\y);
  }
  \foreach \x in {0,1,2,3} {
    \foreach \y in {0,1,...,5} {
      \draw[fill=blue] (\x,\y) circle (2pt);
    }
  }
  \foreach \y in {0,1,2,3,4} {
    \node[anchor=east] at (-0.5,4.5-\y) {$\y$};
  }
  \node[anchor=west] at (3.5,4.5) {$1/2$};
  \node[anchor=west] at (3.5,3.5) {$1/3$};
  \node[anchor=west] at (3.5,2.5) {$1/3$};
  \node[anchor=west] at (3.5,1.5) {$1/5$};
  \node[anchor=west] at (3.5,0.5) {$1/7$};

  \draw[->,shorten >=3pt] (0,5.5) -- (0,5);
  \draw[->,shorten <=3pt] (3,0) -- (3,-0.5);

  \node[rotate=-90,anchor=south] at (4.5,2.5) {probabilities};
  \node[rotate=90,anchor=south] at (-1,2.5) {labels (program points)};
\end{tikzpicture}
\end{center}
\caption{
  Abstract view of the program in \autoref{fig:running-example}.
  Each label on the left identifies a smudge command.
  The dashed, vertical lines signify that once an object is dirty it remains dirty.
  The solid, oblique lines signify that smudge commands might propagate dirtiness.
  Depending on the values of the objects,
    a $\tt smudgeK$ command propagates dirtiness with probability~$1/K$.
  The highlighted path illustrates one way in which dirtiness could
    propagate from object~$x$ to object~$v$,
    thus violating the assertion.
}
\label{fig:running-example-picture}
\end{figure}
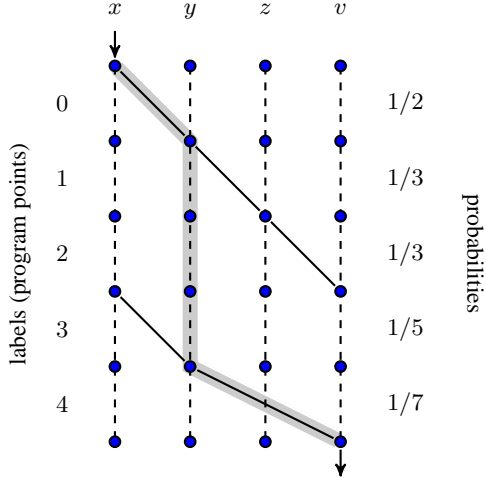
% >>>

However, tracking all values may also be unnecessary.
In the first iteration, the analysis treats all non-smudge commands as $\tt skip$.
As a result, the analysis knows nothing about the {\it value\/} fields.
To remain sound, it assumes that smudge commands always propagate dirtiness;
  that is, it treats the command ${\tt smudgeK}(x,y)$
  as equivalent to the following pseudocode, dropping the guard:
$$
  y.{\it dirty} := x.{\it dirty} \lor y.{\it dirty}
$$
If, using these approximate semantics, the analysis concluded that $v$~is clean at the end,
  then it would stop.
But, in our example,
  $v$~could be dirty at the end,
  for example because of the smudge commands on lines $0$~and~$4$:
  the smudge on line~$0$ propagates dirtiness from~$x$ to~$y$,
  and the smudge on line~$4$ propagates dirtiness from~$y$ to~$v$.
This scenario corresponds to the highlighted path in \autoref{fig:running-example-picture}.

Before seeing what happens in the next iteration,
  let us first describe the analysis in more detail.
%We consider a static analysis implemented in Datalog.
The approximate semantics of the command $\tt smudge2$
  are modelled by the following Datalog rule:
\begin{align}
\begin{aligned}
  {\tt dirty}(\ell', \beta)
    \gets\;
    &{\tt cheap}(\ell),
    {\tt dirty}(\ell,\alpha),
    {\tt flow}(\ell,\ell')\\
    &{\tt smudge2}(\ell,\alpha,\beta)
\end{aligned}
  \label{eq:cheaprule-smudge2}
\end{align}
The rule makes use of the following relations:
\begin{align*}
{\tt flow}(\ell,\ell')\quad
  &\text{the control flow goes from $\ell$ to $\ell'$} \\
{\tt smudge2}(\ell,\alpha,\beta)\quad
  &\text{the command at $\ell$ is ${\tt smudge2}(\alpha,\beta)$}\\
{\tt cheap}(\ell)\quad
  &\text{the command at $\ell$ should be approximated} \\
{\tt dirty}(\ell,\alpha)\quad
  &\text{$\alpha.{\it dirty}$ is $\bf true$ before the command at~$\ell$}
\end{align*}
The relations {\tt flow} and {\tt smudge2} encode the program that is being analysed.
The relation {\tt cheap} parametrises the analysis,
  by allowing it or disallowing it to approximate the semantics of particular commands.
Finally,
  the relation {\tt dirty} expresses facts about executions of the program
  that is being analysed.
From the point of view of the analysis,
  {\tt flow}, {\tt smudge2}, and {\tt cheap} are part of the input,
  while {\tt dirty} is part of the output.
The relations {\tt flow} and {\tt smudge2} are simply a transliteration of the program text.
The relation {\tt cheap} is computed by a refinement algorithm, which we will see later.

The precise semantics of {\tt smudge2} can also be encoded with a Datalog rule,
  albeit a more complicated one.
\begin{align}
\begin{aligned}
  {\tt dirty}(\ell', \beta)
    \gets\;
    &{\tt precise}(\ell),
    {\tt dirty}(\ell,\alpha),
    {\tt flow}(\ell,\ell'),\\
    &{\tt smudge2}(\ell,\alpha,\beta),
    {\tt value}(\ell,\alpha,a),\\
    &{\tt value}(\ell,\beta,b),(a + b) \bmod 2 = 0
\end{aligned}
  \label{eq:preciserule-smudge2}
\end{align}
This rule makes use of two further relations:
\begin{align*}
{\tt precise}(\ell)\quad
  &\text{the command at $\ell$ should not be approximated} \\
{\tt value}(\ell,\alpha,a)\quad
  &\text{$\alpha.{\it value}=a$ holds before the command at $\ell$}
\end{align*}
Like {\tt cheap}, the relation {\tt precise} is part of the input.
If the input relation {\tt precise} activates rules like the one above,
  then the analysis takes longer not only because the rule is more complicated,
  but also because it needs to compute more facts about the relation {\tt value}.

The refinement algorithm ensures that
  for each program point~$\ell$
    exactly one of ${\tt cheap}(\ell)$ and ${\tt precise}(\ell)$ holds.
In the first iteration,
  ${\tt cheap}(\ell)$ holds for all~$\ell$, and $\tt precise$ holds for no~$\ell$.
In each of the next iterations,
  the refinement algorithm switches some program points from cheap to precise semantics.

Let us see what happens when one program point is switched from cheap to precise.
In the first iteration, ${\tt cheap}(0)$ is part of the input,
  and the following rule instance derives ${\tt dirty}(0',{\tt y})$:
\begin{align*}
\begin{aligned}
  {\tt dirty}(0', {\tt y})
    \gets\;
    &{\tt cheap}(0),
    {\tt dirty}(0,{\tt x}),
    {\tt flow}(0,0')\\
    &{\tt smudge2}(0,{\tt x},{\tt y})
\end{aligned}
\end{align*}
Let us now look at the scenario in which for the second iteration
  the fact ${\tt cheap}(0)$ is replaced by the fact ${\tt precise}(0)$.
In this case, ${\tt dirty}(0',{\tt y})$ is still derived,
  this time by the following rule instance:
\begin{align*}
\begin{aligned}
  {\tt dirty}(0', {\tt y})
    \gets\;
    &{\tt precise}(0),
    {\tt dirty}(0,{\tt x}),
    {\tt flow}(0,0'),\\
    &{\tt smudge2}(0,{\tt x},{\tt y}),
    {\tt value}(0,{\tt x},10),\\
    &{\tt value}(0,{\tt y},0),
    (10 + 0) \bmod 2 = 0
\end{aligned}
\end{align*}
To be able to apply this rule, the analysis had to work harder,
  to derive the intermediate results
  ${\tt value}(0,{\tt x},10)$ and ${\tt value}(0,{\tt y},0)$.
Using ${\tt precise}(0)$ influences other Datalog rules as well,
and forces the analysis to derive these intermediate results,
so that ${\tt dirty}(0', {\tt y})$ is still derived.
This is not always the case.
For example,
  the {\tt smudge3} command at program point~$1$ will not propagate dirtiness
  if the precise semantics is used.

Let us now step back and see which parts of the example generalise.

\paragraph{Model.}
If we replace ${\tt cheap}(\ell)$ by ${\tt precise}(\ell)$,
  then the set of Datalog rule instances could change unpredictably.
Yet, we observe empirically that the change is confined to one of two cases:
\begin{enumerate}
\item[(a)]
  ${\tt precise}(\ell)$ eventually derives facts
    similar to those facts that ${\tt cheap}(\ell)$ derives, but with more work; or
\item[(b)]
  ${\tt precise}(\ell)$ no longer derives the facts that ${\tt cheap}(\ell)$ derived.
\end{enumerate}
This dichotomy is by no means necessary.
Intuitively, it holds because the Datalog rules are not arbitrary:
  they are implementing a program analysis.
In our example,
  case~(a) occurs when ${\tt cheap}(0)$ is replaced by ${\tt precise}(0)$,
  and case~(b) occurs when ${\tt cheap}(1)$ is replaced by ${\tt precise}(1)$.
In general,
  we formalise this dichotomy by requiring that a certain predictability condition holds.
The condition is flexible,
  in that it allows one to choose the meaning of `similar' in case~(a)
  by defining a so called projection function.
In our example, no projection is necessary.
In context sensitive analyses, projection corresponds to truncating contexts.
In general,
  by adjusting the definition of the projection function
  we can exploit more knowledge about the analysis, if we so wish.
If we do not, then it is always possible to choose a trivial projection for which the meaning
  of `similar' is `exactly the same'.

Provided that the predictability condition holds,
  which is a formal way of saying that the dichotomy between cases (a)~and~(b) holds,
  it is natural to define the probabilistic model
    as a variant of the Erd\H{o}s--R\'{e}nyi random graph model.
Our sets of Datalog rule instances are seen as sets of arcs of a hypergraph.
Each arc of the hypergraph is either selected or not, with a certain probability.
Being selected corresponds to case~(a) -- having a counterpart in the precise hypergraph;
being unselected corresponds to case~(b) -- not having a counterpart in the precise hypergraph.

For the predictability condition and for the projection function,
  we drew inspiration from abstract interpretation~\cite{Cousot77}.
Intuitively,
  our projection functions correspond to concretisation maps,
  and our predictability condition corresponds to correctness of approximation.
However, we did not formalise this intuitive correspondence.

\paragraph{Learning.}
The model predicts that each rule instance is selected (that is, has a precise counterpart)
  with some probability.
How to pick this probability?
\autoref{fig:running-example-picture} gives an intuitive representation of a set of instances.
In particular, each dashed arc and each solid arc represents some rule instance.
We assume that instances represented by dashed arcs are selected with probability~$1$.
These are instances of some rule which says that a dirty object remains dirty.
We also assume that instances represented by solid arcs are selected with probability~$1/K$.
These are instances of rules of the form~\eqref{eq:cheaprule-smudge2},
  which describe the semantics of ${\tt smudgeK}$ commands.
These probabilities make intuitive sense.
In particular,
  it is reasonable to expect that a number is a multiple of~$K$ with probability~$1/K$.

But, how can we design an algorithm to find these probabilities,
  without appealing to intuition and knowledge about arithmetic?
The answer is that we run the analysis on many programs,
  and observe whether rule instances have precise counterparts or not.
In our example, if the training sample is large enough,
  we would observe that instances of the form~\eqref{eq:cheaprule-smudge2}
  do indeed have counterparts of the form~\eqref{eq:preciserule-smudge2}
  in about $1/K$~of cases.
In general, it is not possible to observe directly which rules have precise counterparts.
It is difficult to decide which rule is a counterpart of which rule.
Instead, we make indirect observations based on which similar facts are derived.
%This complicates the algorithm that learns probabilities,
%  but we have found an efficient solution.

\paragraph{Refinement.}
In terms of \autoref{fig:running-example-picture},
  refinement can be understood intuitively as follows.
We are interested in whether
  there is a path from the input on the top left to the output on the bottom right.
We know the dashed arcs are really present:
  they have a precise counterpart with probability~$1$.
We do not know if the solid arcs are really present:
  we see them only because we used a cheap parameter setting,
  and they have a precise counterpart only with probability~$1/K$.
We can find out whether the solid arcs are really present or just an illusion,
  by running the analysis with a more precise parameter setting.
But, we have to pay a price, because more precise parameter settings are also more expensive.

The question is then which of the solid arcs should we enquire about,
  such that we decide quickly whether there is a path from input to output.
There are several possible strategies,
  in particular there is an optimistic strategy and a pessimistic strategy.
The optimistic strategy hopes that there is no path, so object $v$ is clean at the end.
Accordingly,
  the optimistic strategy considers asking about those sets of solid arcs
  that could disconnect the input from the output,
  if the arcs were not really there.
The pessimistic strategy hopes that there is a path, so object $v$ is dirty at the end.
Accordingly,
  the pessimistic strategy considers asking about those sets of solid arcs
  that could connect the input to the output,
  if the arcs were really there.
The highlighted path in \autoref{fig:running-example-picture}
  corresponds to replacing ${\tt cheap}(0)$ by ${\tt precise}(0)$,
  and also ${\tt cheap}(4)$ by ${\tt precise}(4)$.
Thus, let us denote its set of arcs as~$04$.
There are two other paths that the pessimistic strategy will consider,
  whose sets of arcs are $012$ and $34$.
The path $04$ gets a probability $1/2\times 1/7$ of surviving;
  the path $012$ gets a probability $1/2\times 1/3 \times 1/3$ of surviving;
  the path $34$ gets a probability $1/5\times 1/7$ of surviving.
According to probabilities, the path $04$ has the highest chance of showing
  that $v$~is dirty at the end.

We designed an algorithm which generalises the pessimistic strategy described above
  by taking into account unions of paths
  and also the runtime cost of trying a parameter setting.
Our refinement algorithm has to work in a more general setting
  than suggested by \autoref{fig:running-example-picture}.
In particular, it must handle hypergraphs, not just graphs.

% >>>
% vim:spell spelllang=en_gb:fmr=<<<,>>>:

\section{Preliminaries and Notations}\label{sec:preliminaries} % <<<

In this section we recall several basic notions from probability theory.
At the same time, we introduce the notation used throughout the~paper.
%We shall use simplified definitions, which are sufficient for our purposes.

A \df{finite probability space} is a finite set~$\Omega$
  together with a function $\Pr : \Omega\to\R$ such that
  $\Pr(\omega) \ge 0$ for all $\omega\in\Omega$, and
  $\sum_{\omega\in\Omega}\Pr(\omega)=1$.
An \df{event} is a subset of~$\Omega$.
The \df{probability of an event}~$A$ is
\begin{align*}
  \Pr(A) \defeq \sum_{\omega\in A}\Pr(\omega)
    = \sum_{\omega\in\Omega} \Pr(\omega) [\omega\in A]
\end{align*}
The notation $[\Psi]$ is the Iverson bracket:
  if $\Psi$~is true it evaluates to~$1$,
  if $\Psi$~is false it evaluates to~$0$.
A \df{random variable} is a function ${\bf X} : \Omega\to{\cal X}$.
For each value $x \in {\cal X}$, the set ${\bf X}^{-1}(x)$ is an event,
  traditionally denoted by $({\bf X}=x)$.
In particular, we write $\Pr({\bf X}=x)$ for its probability;
  occasionally, we may write $\Pr(x={\bf X})$ for the same probability.
A \df{boolean random variable} is a function ${\bf X} : \Omega\to\{0,1\}$.
For a random variable~${\bf X}$ with ${\cal X}\subseteq\R$,
  we define its \df{expectation}~$\E {\bf X}$ by
\begin{align*}
  \E {\bf X} \defeq \sum_{x\in{\cal X}} x \Pr({\bf X}=x)
    = \sum_{\omega\in\Omega} \Pr(\omega) {\bf X}(\omega)
\end{align*}
In particular, if ${\bf X}$~is a boolean random variable, then
\begin{align*}
  \E {\bf X} = \Pr({\bf X}=1)
\end{align*}

Events $A_1,\ldots,A_n$ are said to be \df{independent} when
\begin{align*}
  \Pr(A_1\cap\ldots\cap A_n) = \prod_{i = 1}^n \Pr(A_i)
\end{align*}
Note that $n$~events could be pairwise independent, but still dependent when taken altogether.
Random variables ${\bf X}_1,\ldots,{\bf X}_n$ are said to be independent
  when the events $({\bf X}_1=x_1), \ldots,({\bf X}_n=x_n)$ are independent
  for all $x_1,\ldots,x_n$ in their respective domains.
In particular, if ${\bf X}_1,\ldots,{\bf X}_n$ are independent boolean random variables,
  then ${\bf X_1}\land\ldots\land{\bf X}_n$ is also a boolean random variable, and
\begin{align*}
  \E({\bf X}_1\land\ldots\land{\bf X}_n) = \prod_{i = 1}^n \E {\bf X}_i 
\end{align*}
Events $A$~and~$B$ are said to be \df{incompatible} when they are disjoint.
In that case, $\Pr(A \cup B)=\Pr(A)+\Pr(B)$.
In particular, if ${\bf X}_1,\ldots,{\bf X}_n$ are boolean random variables
  such that the events $({\bf X}_1=1),\ldots,({\bf X}_n=1)$ are pairwise incompatible,
  then
\begin{align*}
  \E({\bf X}_1\lor\ldots\lor{\bf X}_n) = \sum_{i = 1}^n \E{\bf X}_i
\end{align*}

% >>>
% vim:spell spelllang=en_gb:

\section{Probabilistic Model}\label{sec:model} % <<<

The probabilistic model predicts what analyses would do
  if they were run with precise parameter settings.
To make such predictions, the model relies on several assumptions:
  the analysis must be implemented in Datalog (\autoref{sec:model-datalog})
  and its precision must be configurable by parameters (\autoref{sec:model-analyses});
  furthermore, increasing precision should correspond to invalidating some derivation steps
    (\autoref{sec:model-hypothesis}).
Given probabilities that individual derivation steps survive the increase in precision,
  we compute probabilities that sets of derivation steps survive the increase in precision
  (\autoref{sec:model-probabilities}).
Given which set of derivation steps survives the increase in precision,
  we can tell whether a given query, which signifies a bug, is still reachable
  (\autoref{sec:model-use}).

\subsection{Datalog Programs and Hypergraphs}\label{sec:model-datalog} % <<<

We shall use a simplified model of Datalog programs, which is essentially a directed hypergraph.
The semantics will then be given by reachability in this hypergraph.
For readers already familiar with Datalog,
  it may help
    to think of vertices as elements of Datalog relations,
    and to think of arcs as instances of Datalog rules with non-relational constraints removed.
For readers not familiar with Datalog,
  simply thinking in terms of the hypergraph introduced below
  will be sufficient to understand the rest of the paper.

We assume a finite universe of \df{facts}.
An \df{arc} is a pair~$(h,B)$ of a head~$h$ and a body~$B$;
  the \df{head} is a fact;
  the \df{body} is a set of facts.
A \df{hypergraph} is a set of arcs.
The \df{vertices} of a hypergraph are those facts that appear in its arcs.
If a hypergraph~$G$ contains an arc $(h,B)$, then we say that $h$~is reachable from~$B$ in~$G$.
In general, given a hypergraph~$G$ and a set~$T$ of facts,
  the set $\reach{G}{T}$ of facts reachable from~$T$ in~$G$
  is defined as the least fixed-point of the following recursive equation:
$$\{\,h\mid\text{$(h,B)\in G$ and $B\subseteq\reach{G}{T}$}\,\} \cup T
  \;\subseteq\; \reach{G}{T}$$  % revert to inline if space needed
The following monotonicity properties are easy to check.
\begin{proposition}\label{prop:monotone-reachability}
  Let $G$, $G_1$ and $G_2$ be hypergraphs;
  let $T$, $T_1$ and~$T_2$ be sets of facts.
\begin{enumerate}
\item[(a)] If $T_1\subseteq T_2$, then $\reach{G}{T_1}\subseteq\reach{G}{T_2}$.
\item[(b)] If $G_1\subseteq G_2$, then $\reach{G_1}{T}\subseteq\reach{G_2}{T}$.
\end{enumerate}
\end{proposition}

Given a hypergraph~$G$ and a set~$T$ of facts,
  the \df{induced sub-hypergraph}~$G[T]$ retains those arcs that mention facts from~$T$:
\begin{align*}
  G[T] \defeq \{\,(h,B)\in G\mid\text{$h\in T$ and $B\subseteq T$}\,\}
\end{align*}

% >>>
\subsection{Analyses}\label{sec:model-analyses} % <<<

We use Datalog programs to implement static analyses that are parametric and monotone.
Thus, the Datalog programs we consider have additional properties:
\begin{enumerate}
\item
  Because the Datalog program implements a static analysis,
  a subset of facts encode queries,
  corresponding to assertions in the program being analysed.
\item
  Because the static analysis is parametric,
  a subset of facts encode parameter settings.
\item
  Because the static analysis is monotone,
  parameter settings that are more expensive are also more precise.
\end{enumerate}
For example, in \autoref{sec:overview},
  queries are facts from the relation {\tt dirty};
  parameter settings are encoded by relations {\tt cheap} and {\tt precise};
  and switching a parameter from {\tt cheap} to {\tt precise}
    makes the analysis more expensive but cannot grow the relation {\tt dirty}.

If we only assume that the analysis is parametric, monotone, and implemented in Datalog,
  then we can already make good predictions in some cases,
  such as the case of the analysis in \autoref{sec:overview}.
In other cases, we require more information about the relationship between
  what the analysis does when run in a precise mode
  and what the analysis does when run in an imprecise mode.
We assume that this information comes in the form of a partial function that projects facts.
The technical requirements on the projection function are mild,
  so the analysis designer has considerable leeway in choosing an appropriate projection.
In some cases, the choice is straightforward.
For example,
  if the analysis is $k$-object sensitive,
  meaning that it tracks calling contexts using sequences of allocation sites,
  then a good choice of projection corresponds to truncating these sequences.

An \df{analysis}~${\cal A}$ is a tuple $(G,Q,P,p_0,p_1,\pi)$, where
  $G$~is a hypergraph called the \df{global provenance},
  $Q$~is a set of facts called \df{queries},
  $P$~is a finite set of \df{parameters},
  the \df{encoding functions} $p_0$~and~$p_1$ map parameters to facts, and
  $\pi$~is a partial function from facts to facts called \df{projection}.
A \df{parameter setting}~$a$ of an analysis~${\cal A}$
  is an assignment of booleans to the parameters~$P$.
We sometimes refer to parameter settings as \df{abstractions}, for brevity.
We encode the abstraction~$a$ as two sets of facts, $P_0(a)$~and~$P_1(a)$, defined by
\begin{align*}
  P_k(a) &\defeq \{\,p_k(x)\mid\text{$x\in P$ and $a(x)=k$}\,\}
    &&\text{for $k\in\{0,1\}$}
\end{align*}
The set ${\cal A}(a)$ of facts \df{derived} by the analysis~${\cal A}$ under abstraction~$a$
  is defined to be $\reach{G}{\bigl(P_0(a)\cup P_1(a)\bigr)}$.
Abstractions form a complete lattice with respect to the pointwise order:
  $a \le a'$ iff $a(x) \le a'(x)$ for all $x\in P$.
We write $\bot$ for the \df{cheapest abstraction} that assigns~$0$ to all parameters,
  and $\top$ for the \df{most precise abstraction} that assigns~$1$ to all parameters.

For an analysis~${\cal A}$,
  we sometimes consider the restriction of its hypergraph
  to those facts derived under a given abstraction~$a$:
  $G^a \defeq G[{\cal A}(a)]$.
In particular,
  $G^{\bot}$ is called the \df{cheap provenance},
  and $G^{\top}$ is called the \df{precise provenance}.

An analysis is \df{well formed} when it obeys further restrictions:
(i)~facts derived under the cheapest abstraction are fixed-points of the projection,
  $\pi(x)=x$ for $x\in{\cal A}(\bot)$,
(ii)~the image of the projection~$\pi$ is included in ${\cal A}(\bot)$,
(iii)~only fixed-points project on queries,
  $\pi^{-1}(q) \subseteq \{q\}$ for $q \in Q$,
(iv)~the encoding functions $p_0$~and~$p_1$ are injective and have disjoint images, and
(v)~projection is compatible with parameter encoding, $\pi\circ p_1=p_0$.
From (i)~and~(ii) it follows that $\pi$~is idempotent.
These conditions are technical:
  they ease the treatment that follows, but do not restrict which analyses can be modelled.

An analysis~${\cal A}$ is said to be \df{monotone}
  when the set of derived queries decreases as a function of the abstraction:
$a \le a'$ implies $\bigl(Q \cap {\cal A}(a)\bigr) \supseteq \bigl(Q\cap {\cal A}(a')\bigr)$.

%In practice, all analyses are well formed and many are monotone.
%In what follows, all analyses are assumed to be both well formed and monotone.

We can now formally define the main problem.

\begin{problem}\label{problem}
Given are a well formed, monotone analysis~${\cal A}$, and a query~$q$ for~${\cal A}$.
Does there exist an abstraction~$a$ such that $q\notin{\cal A}(a)$?
\end{problem}

Because the analysis is monotone,
  $q\in{\cal A}(a)$ for all~$a$ if and only if $q\in{\cal A}(\top)$.
Thus, one way to solve the problem is to check if $q$~is derived by~${\cal A}$
  under the most precise abstraction~$\top$.
However, this is typically too expensive.
Instead, we consider a class of solutions called \df{monotone refinement algorithms}.
A monotone refinement algorithm evaluates the analysis
  for a sequence $a_1\le\cdots\le a_n$ of abstractions.
%Based on empirical observations,
%  we estimate the total running time of such an algorithm to be $c(a_1)+\cdots+c(a_n)$,
%  where $c(a)=\exp(\alpha (\sum_{x\in P}a(x)))$ for some $\alpha>0$. This means
%  that the cost of running an analysis under an abstraction $a$ is exponential
%  in the number of parameters set to $1$ in $a$.
Refinement algorithms terminate when one of two conditions holds:
(i)~$q\notin{\cal A}(a_n)$ or
(ii)~$q\in\reach{G^{a_n}}{\bigl(P_1(a_n)\bigr)}$.
It is easy to see why $q \notin {\cal A}(a_n)$ implies that \autoref{problem} has answer `yes'.
It is less easy to see why $q\in\reach{G^{a_n}}{\bigl(P_1(a_n)\bigr)}$
  implies that \autoref{problem} has answer `no'.
Intuitively,
  this second termination condition says that the query~$q$
  is reachable even if we rely only on precise semantics.
In other words, our abstract counterexample does not actually have any abstract step.
Formally, we rely on the following lemma:

\begin{lemma}\label{lemma:query-impossible}
Let $q$~be a query for a well formed, monotone analysis~${\cal A}$.
If $q\in\reach{G^a}{\bigl(P_1(a)\bigr)}$ for some abstraction~$a$,
  then $q\in{\cal A}(a')$ for all abstractions~$a'$.
\end{lemma}
\begin{proof}
By \autoref{prop:monotone-reachability}(a), 
$q\in\reach{G^a}{\bigl(P_1(a)\bigr)} =
 \reach{G}{\bigl(P_1(a)\bigr)} \subseteq
 \reach{G}{\bigl(P_1(\top)\bigr)}={\cal A}(\top)$.
We conclude by noting that the analysis is monotone.
\end{proof}

% >>>
\subsection{Predictability}\label{sec:model-hypothesis} % <<<

The precise provenance $G^{\top}$
  contains all the information necessary to answer \autoref{problem}.
Unfortunately, the precise provenance $G^{\top}$ is typically very large and hard to compute.
In contrast, the cheap provenance $G^{\bot}$ is typically smaller and easier to compute.
In fact, most refinement algorithms start with the cheapest abstraction, $a_1=\bot$.
Fortunately, we observed empirically that $G^{\top}$~and~$G^{\bot}$ are compatible,
  in a way made precise next.

We begin by lifting the projection~$\pi$ to sets~$T$ of facts as follows:
\begin{align*}
  \pi(T) \defeq \{\,t'\mid\text{$t'=\pi(t)$ and $t\in T$}\,\}
\end{align*}
In particular, if the partial function $\pi$ is not defined for any $t\in T$,
  then $\pi(T)=\emptyset$.
Our empirical observation is that
\begin{align}
  \pi \circ \reach{G^{\top}}{} \circ P_1 &= \reach{H}{} \circ \pi \circ P_1
  &&\text{for some $H\subseteq G^{\bot}$}
  \label{eq:compatible}
\end{align}
An analysis~${\cal A}$ that obeys condition~\eqref{eq:compatible} is said to be \df{predictable}.
A hypergraph~$H$ that witnesses condition~\eqref{eq:compatible}
  is said to be a \df{predictive provenance} of analysis~${\cal A}$.
For a predictable analysis,
  reachability and projection almost commute on the image of~$P_1$,
  except that if projection is done first, then reachability must ignore some arcs.

%All the analyses we tried turned out to be predictable, for a simple choice of projection~$\pi$.
%We do not know a simple theoretical explanation of why it should be so.
%Intuitively, though, condition~\eqref{eq:compatible} makes sense.
%When run with the cheap abstraction,
%  the analysis approximates the semantics of the programming language being analysed.
%Some of these approximations turn out to be correct, and some turn out to be wrong.
%Correspondingly,
%  some arcs of the cheap provenance are retained in the predictive provenance,
%  and some arcs of the cheap provenance are not retained in the predictive provenance.

The inspiration for condition~\eqref{eq:compatible}
  came from the notion of correct approximation, as used in abstract interpretation.
But, it is not the same.
We tested condition~\eqref{eq:compatible}
  on analyses that do not explicitly follow the abstract interpretation framework,
  and we were surprised that it holds.
Then we designed the example analysis from \autoref{sec:overview}
  so that the reason why condition~\eqref{eq:compatible} holds is apparent:
  Datalog rules come in pairs,
    one encoding precise semantics, the other encoding approximate semantics.
But, for real analyses, we could not discern any such simple reason.
Thus, we consider our empirical finding as surprising and intriguing.

Recall that refinement algorithms use two termination conditions:
  $q\notin{\cal A}(a)$ and $q\in\reach{G^a}{\bigl(P_1(a)\bigr)}$.
Predictive provenances
%, which are small compared to precise provenances,
  help us evaluate the termination conditions of refinement algorithms.

\begin{lemma}\label{lemma:predict-termination}
Let ${\cal A}$ be a well formed, monotone analysis.
Let $a$ be an abstraction, and let $H$ be a predictive provenance.
Finally, let $q$ be a query derived by~${\cal A}$ under the cheapest abstraction~$\bot$.
\begin{enumerate}
\item[(a)]
  If $q\notin{\cal A}(a)$,
  then $q\notin\reach{G^{\bot}}{(P_0(a))}$ and $q\notin\reach{H}{(\pi(P_1(a)))}$.
\item[(b)]
  Also,
  $q\in\reach{G^a}{\bigl(P_1(a)\bigr)}$
  if and only if
  $q\in \reach{H}{(\pi (P_1(a)))}$.
\end{enumerate}
\end{lemma}
Part~(a) lets us approximate the termination condition $q\notin{\cal A}(a)$;
part~(b) lets us evaluate the termination condition $q\in\reach{G^a}{\bigl(P_1(a)\bigr)}$.
In both cases, only small parts of the global provenance $G$ are used,
  namely $G^{\bot}$~and~$H$.
The assumption $q\in{\cal A}(\bot)$ is reasonable:
  otherwise the refinement algorithm terminates after the first iteration.
% XXX
\begin{proof}
Assume that $q\in\reach{H}{(\pi(P_1(a)))}$.
We have
\begin{align*}
  & 
  \reach{H}{(\pi(P_1(a)))}
  =
  \pi\bigl( \reach{G^{\top}}{(P_1(a))} \bigr)
  &
  \mbox{by}\ \eqref{eq:compatible}
\\
  &
  q\,{\in}\, \pi\bigl(\reach{G^{\top}}{(P_1(a))}\bigr)
  \Implies
  q\,{\in}\, \reach{G^{\top}}{(P_1(a))}
  &
  \mbox{by}\ \pi^{-1}(q)\,{\subseteq}\, \{q\}
\\
  &\reach{G^{\top}}{(P_1(a))}
   =
  \reach{G^a}{(P_1(a))}
  \subseteq
  {\cal A}(a)
  &
  \mbox{by}\ \hyperref[prop:monotone-reachability]{\mbox{Prop.}~\ref*{prop:monotone-reachability}(a)}
\end{align*}
Putting these together, we conclude that $q\in{\cal A}(a)$.
Using a very similar argument we can show that
  $q\in\reach{G^{\bot}}{(P_0(a))}$ implies $q\in{\cal A}(a)$.
This concludes the proof of part~(a).

The proof of part~(b) is similar.
\end{proof}

\autoref{lemma:predict-termination}
  tells us that we could evaluate termination conditions more efficiently
  if we knew a predictive provenance.
Alas, we do not know a predictive provenance.

% >>>
\subsection{Probabilities of Predictive Provenances}\label{sec:model-probabilities} % <<<

If we do not know a predictive provenance, then a naive way forward is as follows:
  enumerate each possible predictive provenance,
  see what it predicts,
  and take an average of the predictions.
Our model is only marginally more complicated:
  it considers some possible predictive provenances as more likely than others.
On the face of it,
  enumerating all possible predictive provenances takes us back to an inefficient algorithm.
We will see later how to deal with this problem (\autoref{sec:refinement}).
Now, let us define the probabilistic model formally.

The blueprint of the probabilistic model is given by a cheap provenance~$G^{\bot}$.
To each arc $e\in G^{\bot}$,
  we associate a boolean random variable~${\bf S}_e$,
  and call it the \df{selection variable} of~$e$.
Selection variables are independent but may have different expectations.
We partition $G^{\bot}$ into \df{types} $G^{\bot}_1,\ldots,G^{\bot}_t$,
  and we do not require selection variables to have the same expectation
  unless they have the same type.
Each type $G^{\bot}_k$ has an associated \df{hyperparameter}~$\theta_k$:
  if $e\in G^{\bot}_k$, then we say that $e$~has type~$k$,
  and we require that $\E {\bf S}_e=\theta_k$.
Recall that $\E {\bf S}_e=\Pr({\bf S}_e = 1)$.
We define, in terms of the selection variables,
  a random variable~$\randomH$ whose values are predictive provenances,
  by requiring that ${\bf S}_e=[e\in\randomH]$.
%$$
%  (\randomH = H) \
%  \mbox{if and only if}\
%  (e\in H \Iff {\bf S}_e=1 \ \mbox{for all $e \in G^{\bot}$})
%$$
%In other words, $\randomH$~is defined such that ${\bf S}_e=[e\in\randomH]$.
Thus, the probability of a predictive provenance~$H$ is
\begin{align}
  \Pr(\randomH=H)=
  \prod_{k=1}^t
    \theta_k^{|G^{\bot}_k \cap H|}
    (1-\theta_k)^{|G^{\bot}_k \setminus H|}
\end{align}
For example, if all arcs have the same type,
  then the model has only one hyperparameter~$\theta$,
  and $\Pr(\randomH=H)$ is $\theta^{|H|}(1-\theta)^{|G^{\bot} \setminus H|}$.
%If $\theta=1/2$, then all predictive provenances are assigned the probability $2^{-|G^{\bot}|}$.
At the other extreme, if all arcs have their own type,
  then the model has one hyperparameter~$\theta_e$ for each arc $e\in G^{\bot}$,
  and $\Pr(\randomH=H)$ is $\prod_{e\in G^{\bot}}\theta_e^{[e\in H]}(1-\theta_e)^{[e\notin H]}$.

How many types should there be?
Few types could lead to underfitting, many types could lead to overfitting.
In the implementation, we have one type per Datalog rule.
Intuitively,
  this means that we trust the judgement of whoever implemented the analysis.

%This concludes the formal presentation of the probabilistic model.
%But, one question presents itself:
%How should we group arcs into types?
%To see why the answer is important, consider two extreme situations:
%  if all arcs have the same type,
%  then the model is very inflexible, and it will likely underfit empirical data;
%  if each arc has its own type,
%  then the model is very flexible, and it will likely overfit empirical data.
%There is a natural choice for how to define types.
%Recall that arcs are instances of Datalog rules.
%The natural choice is to define types to be sets of instances of the same Datalog rule.
%Intuitively,
%  defining types in terms of Datalog rules
%  amounts to using the same granularity as was deemed appropriate
%    by whoever implemented the analysis in Datalog.

%Finally, recall `all models are wrong, but some are useful'~\cite{Box-stat76}.

% >>>
\subsection{Use of the Model}\label{sec:model-use} % <<<

Before using the probabilistic model in a refinement algorithm,
  we must choose appropriate values for hyperparameters.
This is done offline, in a learning phase (\autoref{sec:learning}).
After learning, each Datalog rule has an associated probability -- its hyperparameter.
%To use the probabilistic model, it is also necessary to know the cheap provenance~$G^{\bot}$.
%Thus, the refinement algorithm will always start
%  by running the analysis under the cheapest abstraction~$\bot$.

After the first invocation of the analysis we know the cheap provenance~$G^{\bot}$,
  which we use as a blueprint for the probabilistic model.
Then, our model predicts whether $q\in\reach{G^a}{(P_1(a))}$,
  where $a$~is some abstraction not yet tried.
Recall that $q\in\reach{G^a}{(P_1(a))}$ is one of the termination conditions.
%  which is one of the two conditions under which refinement algorithms terminate.
The hypergraph $G^a$ is unknown,
  and thus we model it by a random variable~$\randomG^a$.
However, we do know from \autoref{lemma:predict-termination}(b)
  that $q\in\reach{G^a}{(P_1(a))}$ if and only if $q\in\reach{H}{(\pi(P_1(a)))}$.
Thus,
\begin{align*}
  \Pr\bigl(q\in\reach{\randomG^a}{(P_1(a))}\bigr)
  & 
  {}\;=\;
  \Pr\bigl(q\in\reach{\randomH}{(\pi(P_1(a)))}\bigr)
\\&
  \;=\;
  \sum_{\substack{R\\q\in R}} \Pr \bigl( \reach{\randomH}{(\pi(P_1(a)))} = R \bigr)
\end{align*}
where $R$ ranges over subsets of vertices of $G^\bot$.
It remains to compute a probability of the form $\Pr\bigl(\reach{\randomH}{T}=R\bigr)$.
Explicit expressions for such probabilities are also needed during learning,
  so they are discussed later (\autoref{sec:learning}).

Intuitively,
  one could think that the refinement algorithm runs a simulation
  in which the static analyser is approximated by the probabilistic model.
However, it would be inefficient to actually run a simulation,
  and we will have to use heuristics that have a similar effect (\autoref{sec:refinement}),
  namely to minimise the expected total runtime.

% >>>
% >>>
% vim:spell spelllang=en_gb:fmr=<<<,>>>:

\section{Learning}\label{sec:learning} % <<<

The probabilistic model (\autoref{sec:model}) lets us compute the probability
  that a given abstraction will provide a definite answer, and thus terminate the refinement.
These probabilities are computed as a function of hyperparameters.
The values of the hyperparameters, however, remain to be determined.
To find good hyperparameters,
  we shall use a standard method from machine learning,
  namely MLE (maximum likelihood estimation).

MLE works as follows.
First, we set up an experiment.
The result of the experiment is that we observe an event~$O$.
Next, we compute the \df{likelihood}~$\Pr(O)$ according to the model,
  which is a function of the hyperparameters.
Finally, we pick for hyperparameters values that maximise the likelihood.

The standard challenge in deploying the MLE method is in the last phase:
  the likelihood is typically a complicated function of the hyperparameters.
Often, to maximise the likelihood,
  analytic methods do not exist, and numeric methods could be unstable or inefficient.
This is indeed the case for our model:
  analytic methods do not apply, and many numeric methods are inefficient.
But, we did find one numeric method that is both stable and efficient
  (\autoref{section:experiments-optimisation}).
In addition to the standard challenge, our setting presents an additional difficulty.
%  both of them related to the computation of the likelihood.
%Usually, $O$~has the form $O_1\cap\ldots\cap O_n$,
%  where the events $O_1,\ldots,O_n$ are independent.
%In our setting, the event $O$ does indeed have the form $O_1\cap\ldots\cap O_n$,
%  but the events $O_1,\ldots,O_n$ are not independent.
%We will handle this difficulty
%  by finding a way to compute $\Pr(O)$ other than the factorisation $\prod_{i = 1}^n \Pr(O_i)$.
%Even so, a second difficulty will arise:
The expression of~$\Pr(O)$ is exponentially large
  if the cheap provenance has cycles.
We will handle this difficulty by finding bounds that approximate $\Pr(O)$.
%In short, the difficulties we need to overcome are:
%(1)~there are dependencies among $O_1,\ldots,O_n$, which we need to account for, and
%(2)~the cycles of the cheap provenance make it infeasible to compute the likelihood exactly.

\subsection{Training Experiment}\label{section:learning-experiment} % <<<

For the training experiment, we collect a set of programs.
For the formal development,
  it is convenient to consider the set of programs as one larger program.
We run the analysis on this large training program several times,
  each time under a different abstraction.
The abstractions $a_1,\ldots,a_n$ are chosen randomly, with bias.
In particular,
  they have to be cheap enough so that the analysis terminates in reasonable time.
As a result of running the analysis, we observe the provenances $G^{a_1},\ldots,G^{a_n}$.
To connect these observed provenances to a probabilistic event,
  we shall use the predictability condition~\eqref{eq:compatible}
  together with the following simple fact.

\begin{proposition}\label{prop:hypergraph-slice}
Let $G$ be a hypergraph, and let $T_1$~and~$T_2$ be sets of facts.
If $T_1\subseteq T_2$, then $\reach{G}{T_1}=\reach{G'}{T_1}$,
  where $G'=G[\reach{G}{T_2}]$.
\end{proposition}
\begin{corollary}\label{corollary:provenance-slice}
Let $a$ be an abstraction for analysis~${\cal A}$.
We have $\reach{G^{\top}}{(P_1(a))}=\reach{G^a}{(P_1(a))}$.
\end{corollary}

Given an efficient way to compute the projection~$\pi$,
  we can compute the sets of facts $R_k \defeq \pi\bigl(\reach{G^{a_k}}{(P_1(a_k))}\bigr)$,
  for each $k\in\{1,\ldots,n\}$.
Using \autoref{corollary:provenance-slice} and condition~\eqref{eq:compatible},
  we have that $R_k=\reach{H}{(\pi(P_1(a_k)))}$,
  for $k\in\{1,\ldots,n\}$.
We define the following events:
\begin{align*}
  O_k \;&\defeq\; \bigl( \reach{\randomH}{(\pi(P_1(a_k)))} = R_k \bigr)
    &&\text{\rm for $k\in\{1,\ldots,n\}$}
\\
  O \;&\defeq\; \bigl( O_1 \cap \ldots \cap O_n \bigr)
\end{align*}
The event $O$ is what we observe.
It is completely described by the pairs $(a_k, R_k)$.
The abstraction $a_k$ is sampled at random.
The set~$R_k$ of facts is easily computed from~$G^{a_k}$.
The provenance $G^{a_k}$ is obtained from the set of instantiated Datalog rules during
the analysis under abstraction~$a_k$, and it records all the reasoning steps of the analysis.

%In practice,
%  we do not analyse all training programs at once, but rather one at a time.
%This optimisation is straightforward, and does not warrant additional explanation.

% >>>
\subsection{Bounds on Likelihood}\label{section:learning-likelihood} % <<<

There appears to be no formula that computes the likelihood~$\Pr(O)$
  and that is not exponentially large.
However, there exist reasonably small formulas that provide lower and upper bounds.
We shall use the lower bound for learning,
  and we shall use both bounds to evaluate the quality of the model.

%To state the main result on likelihood computation, we need to define forward arcs.
One could define different bounds on likelihood.
Our choice relies on the concept of forward arc,
  which leads to several desirable properties we will see later.
Given a hypergraph~$G$,
  we define the \df{distance}~$d^{(G)}_T(h)$ from vertices~$T$ to vertex~$h$
  by requiring $d^{(G)}_T$ to be the unique fixed-point of the following equations:
\begin{align*}
d^{(G)}_T(h) & {} = 0 && \mbox{if $h \in T$}
\\
d^{(G)}_T(h) & {} = \infty && \mbox{if $h \not\in \reach{G}{T}$}
\\
d^{(G)}_T(h) & {} = \min_{e = (h,B) \in G} \max_{b \in B} (d^{(G)}_T(b) + 1) && \mbox{otherwise}
\end{align*}
We omit the superscript when the hypergraph is clear from context.
A \df{forward arc} with respect to~$T$ is an arc $e = (h,B) \in G$
  such that $d_T(h) > d_T(b)$ for every $b \in B$.

\begin{restatable}{theorem}{learningtheorem}\label{theorem:learning}
Consider the probabilistic model associated with the cheap provenance~$G^{\bot}$
  of some analysis~${\cal A}$.
Let $T_1,\ldots,T_n$ and $R_1,\ldots,R_n$ be subsets of vertices of~$G^{\bot}$.
If $h\notin B$ for all arcs $(h,B)$ in~$G^{\bot}$
and $R_k \subseteq \reach{G^\bot}{T_k}$ for all $k$,
  then we have the following lower and upper bounds on
    $\Pr\bigl(\bigcap_{k=1}^n ( \reach{\randomH}{T_k} = R_k)\bigr)$:
\begin{align*}
& \prod_{e\in N} \E \bar{\bf S}_e
  \prod_{\substack{h\\C_h \neq \emptyset}}
  \hskip-0.5em
    \sum_{\substack{
      E_1\\
      E_1\subseteq A_h\\
      \forall k\in C_h,\; E_1\cap {\color{blue}F_k}\ne\emptyset
    }}
    \hskip-1em
      \prod_{e\in E_1} \E {\bf S}_e
      \prod_{e\in A_h\setminus E_1} \E \bar{\bf S}_e
\\%[1ex]
  {} \;\le\;
&
  \Pr \Bigl( \bigcap_{k=1}^n \bigl( \reach{\randomH}{T_k} = R_k \bigr) \Bigr)
\\[1ex]
  \;\le\; {}
&
  \prod_{e\in N} \E \bar{\bf S}_e
  \prod_{\substack{h\\C_h \neq \emptyset}}
  \hskip-0.5em
    \sum_{\substack{
      E_1\\
      E_1\subseteq A_h\\
      \forall k\in C_h,\; E_1\cap {\color{blue}D_k}\ne\emptyset
    }}
    \hskip-1em
      \prod_{e\in E_1} \E {\bf S}_e
      \prod_{e\in A_h\setminus E_1} \E \bar{\bf S}_e
\end{align*}
where
\begin{align*}
  N &\defeq
   \{\,(h',B')\in G^{\bot}\mid
    \text{\rm $B'\subseteq R_{k'}$ and $h'\notin R_{k'}$ for some~$k'$}\,\}
\\
  C_h &\defeq \{\,k'\mid h\in R_{k'}\setminus T_{k'}\,\}
  \qquad
  A_h \defeq \{\,(h,B')\in G^{\bot}\,\} \setminus N
\\
  D_k &\defeq \{\,(h',B')\in G^{\bot}\mid B'\subseteq R_k\,\}
\\
  F_k &\defeq \{\,e=(h',B') \in D_k \mid \text{\rm $e$ is a forward arc w.r.t.~$T_k$}\,\}
\end{align*}
\end{restatable}

%The proof is in
%  \longshort{\autoref{section:lemma-learning-proof}}{the full version}.
Intuitively, the arcs in~$N$ are those arcs that were observed to be not selected;
  thus, the factor $\prod_{e \in N} \E \bar{\bf S}_e$.
For each reachable vertex, there is a factor that requires a justification,
  in terms of other reachable vertices and in terms of selected arcs.
Let us consider a simple example, in which the lower and upper bounds coincide:
  there are four arcs $e_k=(h,\{b_k\})$ for $k\in\{1,2,3,4\}$,
  and we observed $R_1=\{b_1\}$,\; $R_2=\{b_1,b_2,b_4,h\}$, and $R_3=\{b_3,b_4,h\}$.
In $R_1$, vertex $h$ is not reachable but $b_1$ is,
  so ${\bf S}_{e_1}$ must not hold.
In $R_2$, vertex $h$ is reachable and could be justified by one of $e_1,e_2,e_4$,
  so ${\bf S}_{e_1} \lor {\bf S}_{e_2} \lor {\bf S}_{e_4}$ must hold.
In $R_3$, vertex $h$ is reachable and could be justified by one of $e_3,e_4$,
  so ${\bf S}_{e_3} \lor {\bf S}_{e_4}$ must hold.
In all,
\begin{align}
\begin{aligned}
  &\bar{\bf S}_{e_1}
  \land ({\bf S}_{e_1} \lor {\bf S}_{e_2} \lor{\bf S}_{e_4})
  \land ({\bf S}_{e_3} \lor{\bf S}_{e_4})
\\
  &\quad=
  \bar{\bf S}_{e_1}
  \land ({\bf S}_{e_2} \lor{\bf S}_{e_4})
  \land ({\bf S}_{e_3} \lor{\bf S}_{e_4})
\end{aligned}
  \label{eq:example-likelihood}
\end{align}
must hold.
The expectation of this quantity is written in \autoref{theorem:learning} as
$
  \E \bar{\bf S}_{e_1}
  (
  \E \bar{\bf S}_{e_2} \E \bar{\bf S}_{e_3} \E {\bf S}_{e_4}
  + \cdots +
  \E {\bf S}_{e_2} \E {\bf S}_{e_3} \E {\bf S}_{e_4}
  )
$,
where the inner sum enumerates the models of
  $({\bf S}_{e_2} \land {\bf S}_{e_3}) \lor {\bf S}_{e_4}$.

The situation becomes more complicated when the hypergraph has cycles.
In the presence of cycles,
  the recipe from the previous example does not compute the likelihood,
  but it does compute an upper bound.
The reason is that it counts all cyclic justifications as if they were valid.
Indeed, this is the upper bound given in \autoref{theorem:learning}.
For the lower bound,
  we first eliminate cycles by dropping some arcs, thus lowering the likelihood;
  then, we apply the same recipe.
\autoref{theorem:learning} indicates that the arcs which should be dropped
  are the nonforward arcs.
Why is this a good choice?
One might think that we should drop a minimum number of arcs if we want a good lower bound.
However,
(1)~it is NP-hard to find the minimum number of arcs%
  ~\cite[Feedback Arc Set]{karp21}, and
(2)~the set of such arcs is not uniquely determined.
In contrast, we can find the set of nonforward arcs in polynomial time,
  and the solution is unique.

Another nice property of the set of forward arcs is that,
  if for each reachable vertex~$h$ we retain at least one forward arc whose head is~$h$,
  then all reachable vertices remain reachable.
This property is desirable for detecting impossibility
  (see \autoref{lemma:feasible-subset}).
In terms of the lower bound,
  this property means that we never lower bound a positive probability by~$0$.

In the implementation, we sometimes heuristically drop forward arcs,
  in order to keep the size of the formula small.
But, we only choose to drop a forward arc with head~$h$
  if there are more than $8$~forward arcs with head~$h$.
For example, if we drop arc~$e_2$ in our running example,
  the effect is that we lower bound~\eqref{eq:example-likelihood} by
$$
  \bar{\bf S}_{e_1}
  \land {\bf S}_{e_4}
  \land ({\bf S}_{e_3} \lor{\bf S}_{e_4})
$$
We simply drop the corresponding variable ${\bf S}_{e_2}$ from the formula,
  thus making the formula smaller.
%  especially if we further apply algebraic simplifications.
Similarly, we can reduce the size of the formula for the upper bound,
  at the cost of weakening the bond.
This time, we drop clauses rather than variables.
For example, we can upper bound~\eqref{eq:example-likelihood} by
$$
  \bar{\bf S}_{e_1}
  \land ({\bf S}_{e_2} \lor{\bf S}_{e_4})
$$
For each vertex,
  our implementation drops all clauses except for the longest one.

\smallskip
Although the probabilistic model is simple,
  computing the likelihood of an event of the form
  `$\reach{\randomH}{T_1}=R_1$ and $\ldots$ and $\reach{\randomH}{T_n}=R_n$'
  is not computationally easy.
\longshort{\autoref{section:lemma-learning-proof}}{The full version of this paper}
  gives an exact formula
  that has size exponential in the number of vertices of the cheap provenance,
  but also points to evidence that a significantly smaller formula is unlikely to exist%
    \longshort{}{~\cite{razborov:many-loops}}.
The size explosion is caused mainly by the cycles of the cheap provenance.
\autoref{theorem:learning} gives likelihood lower and upper bounds
  that are exponential only in the maximum in-degree of the cheap provenance.
These formulas are still too large to be used in practice.
However, there are simple heuristics that can be applied to reduce the size of the formulas,
  at the cost of weakening the bounds.

We use the lower bound to learn hyperparameters
  (\autoref{section:experiments-optimisation}).
We use the upper bound to measure the quality of the learnt hyperparameters
  (\autoref{section:experiments-prediction}).

% >>>
\subsection{Results}\label{section:learning-results} % <<<

We learnt hyperparameters for a flow insensitive but object sensitive aliasing analysis.
%  (see \autoref{section:experiments} for details).
The aliasing analysis is implemented in $59$~Datalog rules.
All but $5$ rules get a hyperparameter of~$1$.
A rule with a hyperparameter of~$1$
  is a rule that was not observed to be involved in any approximation, in the training set.
For two of the remaining five rules,
  the learnt hyperparameters were essentially random,
  because the likelihood lower bound did not depend on them.
The reason is
  that the training set did not contain enough data,
  or that the lower bound was too weak.
% The two rules that got random hyperparam values:
% reachableT(t) :- reachableCI(_,i), objNewInstIH(i,h), HT(h,t).
% reachableT(t) :- reachableCM(_,m), MputStatFldInst(m,f,_), staticTF(t,f).

For the remaining three rules the hyperparameters were $0.997$, $0.985$, and $0.969$.
These values were robust, in the sense that they varied little when the training subset changed.
% The three rules that <1:
% 0.997:  CVC(c,l,c2) :- CVC_33_0(c,l,f,c1), CFC(c1,f,c2).
% 0.985:  CFC(c1,f,c2) :- CFC_37_0(c,f,r,c1), CVC(c,r,c2).
% 0.969:  CVC(c,u,o) :- DVDV(c,u,d,v), CVC(d,v,o), VCfilter(u,o).
For example, the rule with a hyperparameter of~$0.969$ is
$$
  {\tt CVC}(c,u,o) \gets {\tt DVDV}(c,u,d,v), {\tt CVC}(d,v,o), {\tt VCfilter}(u,o)
$$
Looking briefly at the aliasing analysis implementation we see that
  (a)~${\tt CVC}(c,u,o)$ means `in context $c$, variable~$u$ may point to object~$o$', and
  (b)~the relation ${\tt DVDV}$ is responsible for copying method arguments and returned values.
We interpret this as evidence that the approximations done by the aliasing analysis
  are closely related to approximations of the call graph.

We are not the authors of the aliasing analysis; it is taken from Chord.
Our learning algorithm automatically identified the three rules that are most interesting,
  from the point of view of approximation.

% >>>
% >>>
% vim:spell spelllang=en_gb:fmr=<<<,>>>:

\section{Refinement}\label{sec:refinement} % <<<

The probabilistic model is interesting from a theoretical point of view
  (\autoref{sec:model}).
The learning algorithm is already useful,
  because it lets us find which rules of a static analysis approximate the concrete semantics,
  and by how much
  (\autoref{sec:learning}).
In this section we explore another potential use of the learnt probabilistic model:
  to speed up the refinement of abstractions.

We consider a refinement algorithm that is applicable to analyses implemented in Datalog
  (\autoref{section:refinement-algorithm}).
The key step of refinement is choosing the next abstraction to try.
Abstractions that make good candidates share several desirable properties.
In particular,
  they are likely to answer the posed query (\autoref{section:refinement-termination}),
  and they are likely to be cheap to try (\autoref{section:refinement-balance}).
These two desiderata need to be balanced (also \autoref{section:refinement-balance}).
Once we formalise how desirable an abstraction is,
  the next task is to search for the most desirable one
  (\autoref{section:refinement-maxsat}).

\subsection{Refinement Algorithm}\label{section:refinement-algorithm} % <<<

\begin{figure} % <<< main algo
{\sffamily\small
\begin{alg}
Given: A well formed, monotone analysis ${\cal A}$, and a query $q$.
\proc{Solve}
\0  $a := \bot$ \comment $\bot$ as initial abstraction
\0  repeat
\1    $G^a := G[{\cal A}(a)]$ \comment invokes analysis
\1    if $q \notin {\cal A}(a)$ then return ``yes''
\1    if $q \in \reach{G^a}{(P_1(a))}$ then return ``no''
\1    $a := \proc{ChooseNextAbstraction}(G^a, q, a)$
\end{alg}
}
\caption{
  The refinement algorithm used to solve \autoref{problem}.
}
\label{figure:refinement-algorithm}
\end{figure}
% >>>

The refinement algorithm is straightforward (\autoref{figure:refinement-algorithm}).
It repeatedly obtains the provenance~$G^a$
%RG: i think "analysis A *under* abstraction a" is used in many places
by running the analysis under abstraction~$a$ (line~3),
  checks if one of the two termination conditions holds (lines 4~and~5),
  and invokes $\textsc{ChooseNextAbstraction}$ to update the current abstraction (line~6).
The correctness of this algorithm follows
  from the discussion in \autoref{sec:model-analyses},
  and in particular \autoref{lemma:query-impossible}.

Let $a'$ be the result of $\textsc{ChooseNextAbstraction}(G^a,q,a)$.
For termination, we require that $a'$~is strictly more precise than~$a$.
This is sufficient because the lattice of abstractions is finite.
The next abstraction to try should satisfy two further requirements:
\begin{enumerate}
\item The termination conditions are likely to hold for~$a'$.
\item The estimated runtime of~${\cal A}$ under~$a'$ is small.
\end{enumerate}
Next, we discuss these two requirements in turn.
To some degree, we will make each of them more precise.
But, we caution that from now on the discussion leaves the realm of hard theoretical guarantees,
  and enters the land of heuristic reasoning,
  where discussions about static program analysis are typically found.

% >>>
\subsection{Making Termination Likely}\label{section:refinement-termination} % <<<

The key step of the refinement algorithm (\autoref{figure:refinement-algorithm})
  is the procedure $\textsc{ChooseNextAbstraction}$.
The simplest implementation that would ensure correctness is the following:
  return a random element from the set of feasible abstractions
%\begin{align*}
  $\{\,a'\mid a'>a\,\}$.
%\end{align*}
Note that if $a$~were the most precise abstraction
  then the procedure $\textsc{ChooseNextAbstraction}$ would not be called,
  so the feasible set from above is indeed guaranteed to be nonempty.

One idea to speed up refinement is to restrict the set of feasible solutions
  to those abstractions that are likely to provide a definite answer.
Let $A_{\rm y}$~and~$A_{\rm n}$ be the sets of abstractions
  that will lead the refinement algorithm to terminate on the next iteration
  with the answer `yes' or, respectively, `no':
\begin{align*}
  A_{\rm y} \;&\defeq\; \{\,a'\mid\text{$a'>a$ and $q\notin{\cal A}(a')$}\,\} \\
  A_{\rm n} \;&\defeq\; \{\,a'\mid\text{$a'>a$ and $q\in\reach{G^{a'}}{(P_1(a'))}$}\,\}
\end{align*}
Of course, exactly one of the two sets $A_{\rm y}$~and~$A_{\rm n}$ is nonempty,
  but we do not know which.
More generally, we cannot evaluate these sets exactly without running the analysis.
But, we can approximate them, because $\textsc{ChooseNextAbstraction}$ has access to~$G^a$.
For $A_{\rm y}$ we can compute an upper bound~$A_{\rm y}^{\supseteq}$;
  for $A_{\rm n}$ we use a heuristic approximation~$A_{\rm n}^{\approx}$.
\begin{align*}
  A_{\rm y}^{\supseteq} \;&\defeq\;
    \{\,a'\mid\text{$a'>a$ and $q\notin\reach{G^a}{(P_0(a'))}$}\,\} \\
  A_{\rm n}^{\approx} \;&\defeq\;
    \{\,a'\mid\text{$a'>a$ and $q\in\reach{H}{(T(a,a'))}$}\,\}
\intertext{for some $H \subseteq G^a$, where}
  T(a,a') \;&\defeq\; P_1(a) \cup \pi(P_1(a') \setminus P_1(a))
\end{align*}
It is easy to see why $A_{\rm y}^{\supseteq}\supseteq A_{\rm y}$;
  it is less easy to see why $A_{\rm n}^{\approx}\approx A_{\rm n}$.
Let us start with the easy part.

\begin{lemma}
Let $A_{\rm y}^{\supseteq}$~and~$A_{\rm y}$ be defined as above.
Then $A_{\rm y}^{\supseteq} \supseteq A_{\rm y}$.
\end{lemma}
% XXX
\begin{proof}
Assume that $a'>a$, as in the definitions of $A_{\rm y}^{\supseteq}$~and~$A_{\rm y}$.
Then $P_0(a')\subseteq P_0(a)$.
By \autoref{prop:hypergraph-slice} and \autoref{prop:monotone-reachability},
\begin{align*}
  \reach{G^a}{(P_0(a'))}
  =\reach{G}{(P_0(a'))}
  =\reach{G^{a'}}{(P_0(a'))}
  \subseteq{\cal A}(a')
\end{align*}
The claimed inclusion now follows.
\end{proof}

Let us now discuss the less obvious claim that $A_{\rm n}^{\approx}\approx A_{\rm n}$.
One could wonder why we did not define $A_{\rm n}^{\approx}$ by
\begin{align*}
  \{\,a'\mid\text{$a'>a$ and $q\in\reach{H}{(\pi(P_1(a')))}$}\,\}
\end{align*}
for some $H \subseteq G^{\bot}$.
This definition is simpler and is also guaranteed to be equivalent to $A_{\rm n}$,
  by the predictability condition~\eqref{eq:compatible}.
In the implementation,
  we use the more complicated definition of $A_{\rm n}^{\approx}$ for two reasons.
First,
  we note that \eqref{eq:compatible}~implies $A_{\rm n}^{\approx}=A_{\rm n}$
  if $a=\bot$.
Thus,
  the claim that $A_{\rm n}^{\approx}=A_{\rm n}$
  can be seen as a generalisation of~\eqref{eq:compatible}.
We did not use this generalisation of~\eqref{eq:compatible}
  in the more theoretical parts (\autoref{sec:model} and \autoref{sec:learning})
  because it would complicate the presentation considerably.
For example, instead of one projection~$\pi$,
  we would have a family of projections that compose.
In principle, however, it would be possible to take $A_{\rm n}^{\approx}=A_{\rm n}$
  as an axiom, from the point of view of the theoretical development.
Second, the more complicated definition of~$A_{\rm n}^{\approx}$
  exploits all the information available in~$G^a$.
The simpler version can also incorporate information from~$G^a$
  by conditioning $H$ to be compatible with~$G^a$, via~\eqref{eq:compatible}.
However, this conditioning would only use the projected set of vertices of~$G^a$,
  rather than its full structure.

Furthermore,
  the definition of~$A_{\rm n}^{\approx}$ used in the implementation
  has the following intuitive explanation.
The condition $A_{\rm n}^{\approx}\approx A_{\rm n}$ tells us that
  in order to predict $\reach{G^{a'}}{(P_1(a'))}$ by using~$G^a$ we should do the following:
(i)~split $P_1(a')$ into $P_1(a)$ and $P_1(a') \setminus P_1(a)$;
(ii)~use the facts $P_1(a)$ as they are, because they already appear in~$G^a$;
(iii)~approximate the facts in $P_1(a') \setminus P_1(a)$ by their projections,
  because they do not appear in~$G^a$; and
(iv)~define the predictive provenance~$H$ with respect to~$G^a$,
  because it is the most precise provenance available so far.

We defined two possible restrictions of the feasible set,
  namely $A_{\rm y}^{\supseteq}$~and~$A_{\rm n}^{\approx}$.
The remaining question is now which one should we use,
  or whether we should use some combination of them such as
    $A_{\rm y}^{\supseteq} \cap A_{\rm n}^{\approx}$.
The restriction to~$A_{\rm y}^{\supseteq}$ could be called the optimistic strategy,
  because it hopes the answer will be `yes';
the restriction to~$A_{\rm n}^{\approx}$ could be called the pessimistic strategy,
  because it hopes the answer will be `no'.
The optimistic strategy has been used in previous work~\cite{zhang:pldi14}.
The pessimistic strategy is used in our implementation.
We found that it leads to smaller runtime (\autoref{section:experiments-runtime}).
It would be interesting to explore combinations of the two strategies,
  as future work.

In the optimistic strategy, one needs to check whether $A_{\rm y}^{\supseteq}=\emptyset$.
In this case, it must be that $A_{\rm y}=\emptyset$ and thus the answer is `no'.
In other words,
  the main loop of the refinement algorithm needs to be slightly modified
  to ensure correctness.
In the pessimistic strategy, it is never the case that $A_{\rm n}^{\approx}=\emptyset$,
  and so the main loop of the refinement algorithm is correct as given
  in \autoref{figure:refinement-algorithm}.
The pessimistic restriction $A_{\rm n}^{\approx}$~is nonempty
  because it always contains~$\top$,
  by choosing $H=G^a$ (see \autoref{lemma:feasible-subset}).

The set $A_{\rm n}^{\approx}$ is defined in terms of an unknown predictive provenance~$H$.
Thus, we work in fact with the random variable
\begin{align*}
  {\bf A}_{\rm n}^{\approx} &\defeq
    \{\,a'\mid\text{$a'>a$ and
      $q\in\reach{\randomH}{(T(a,a'))}$}\,\}
\end{align*}
defined in a probabilistic model with respect to~$G^a$, instead of~$G^{\bot}$.
We wish to choose an abstraction~$a'$ that is likely in~${\bf A}_{\rm n}^{\approx}$.
In other words, we want to maximise $\Pr(a' \in {\bf A}_{\rm n}^{\approx})$.
There is no simple expression to compute this probability.
For optimisation, we will use the following lower bound.

\begin{restatable}{lemma}{noproblemma}\label{lemma:noprob-lowerbound}
Let ${\bf A}_{\rm n}^{\approx}$ be defined as above,
  with respect to an analysis~${\cal A}$, an abstraction~$a$, and a query~$q$.
Let $a'$ be some abstraction such that $a'>a$.
Let $H$ be some subgraph of~$G^a$ such that $q \in \reach{H}{(T(a,a'))}$.
Then
\begin{align*}
  \Pr(a' \in {\bf A}_{\rm n}^{\approx}) \;\ge\; \prod_{e \in H} \E {\bf S}_e
\end{align*}
where ${\bf S}_e$~is the selection variable of arc~$e$.
\end{restatable}

Before describing the search procedure (\autoref{section:refinement-maxsat}),
  we must see how to balance
    maximising the probability of termination with minimising the running cost.

% >>>
\subsection{Balancing Probabilities and Costs}\label{section:refinement-balance} % <<<

We are looking for an abstraction that is likely to answer the query but,
  at the same time, is not too expensive.
Most of the time, these two desiderata point in opposite directions:
  expensive abstractions are more likely to provide an answer.
This raises the question of how to balance the two desiderata.
We model the problem as follows.

\begin{definition}[Action Scheduling Problem]\label{problem:scheduling}
Suppose that we have a list of $m \geq 1$ actions, which can succeed or fail.
The success probabilities
of these actions are $p_1,\ldots,p_m \in (0,1]$,
and the costs for executing these actions are
$c_1,\ldots,c_m > 0$. Find a permutation $\sigma$ on $\{1,\ldots,m\}$ that minimises
the cost $C(\sigma)$:
\begin{align*}
  C(\sigma) = \sum_{k=1}^{m} q_k(\sigma) c_{\sigma(k)},
\qquad
  q_k(\sigma) = \prod_{j=1}^{k-1} \bigl(1-p_{\sigma(j)}\bigr).
\end{align*}
Intuitively, $C(\sigma)$ represents the average
cost of running actions according to $\sigma$ until we hit success.
\end{definition}
In the setting of our algorithm,
  the $m$~actions correspond to all the possible next abstractions $a'_1,\ldots,a'_m$.
The $p_i$ is $\Pr(a'_i \in {\bf A}_{\rm n}^{\approx})$,
  and $c_i$~is the cost of running the analysis under abstraction~$a'_i$.
Hence, a solution to this action scheduling problem tells us how we should
combine probability and cost, and select the next abstraction~$a'$.

%We prove that under some independence assumption,
%we can solve the action scheduling problem:
\begin{restatable}{lemma}{greedylemma}\label{lemma:greedy}
Consider an instance of the action scheduling problem (\autoref{problem:scheduling}).
Assume the success probabilities of the actions are independent. A
permutation $\sigma$ has minimum cost~$C(\sigma)$ if and only if
  $p_{\sigma(1)}/c_{\sigma(1)} \ge \cdots \ge p_{\sigma(m)}/c_{\sigma(m)}$.
%for all $1 \leq i,j \leq m$,
%\begin{align}
%\label{eqn:greedy1}
%i \le j \quad\Implies\quad
%  p_{\sigma(i)} / c_{\sigma(i)} \ge p_{\sigma(j)} 
%  \frac{c_{\sigma(i)}}{p_{\sigma(i)}} \le \frac{c_{\sigma(j)}}{p_{\sigma(j)}}.
%\end{align}
\end{restatable}

%For the proof, see
%  \longshort{\autoref{section:proofs-refinement}}{the submitted supplement}.

\begin{corollary}\label{corollary:greedy-choice}
Under the conditions of \autoref{lemma:greedy},
  if the cost of permutation~$\sigma$ is minimum, then $\sigma(1) \in\argmax_i p_i/c_i$.
\end{corollary}

%>>>
\subsection{\texorpdfstring{$\maxsat$}{MaxSAT} encoding}\label{section:refinement-maxsat} % <<<
% floats <<<
\begin{table}\centering\footnotesize
\begin{tabular}{@{}rccc@{}}
\toprule
Cases & {\tt all-one} & {\tt fine} & {\tt coarse} \\
\midrule
$95.0\%$ & $0$ & $(-0.22,-0.20)$ & $(-0.73,-0.72)$ \\
$3.8\%$ & $-\infty$ & $(-15,-14)$ & $(-33,-32)$ \\
$1.2\%$ & $-\infty$ & $-\infty$ & $(-12,-11)$
\end{tabular}
\caption{
  Bounds on the average log-likelihood, in base~$e$.
% (Bigger is better.)
% The {\tt all-one} model sets all hyperparameters to~$1$.
% The {\tt coarse} model uses one value for all rules.
% The {\tt fine} model uses one value per rule.
% The training set and the evaluation set each has $260$~samples.
}
\label{table:average-likelihood}
\end{table}

\begin{table}\centering\footnotesize
\begin{tabular}{@{}llrrr@{}}
\toprule
\multicolumn{2}{c}{Configuration} & \multicolumn{2}{c}{Solved queries} \\
\cmidrule(r){1-2} \cmidrule(lr){3-4}
Strategy & Optimiser & Ruled out & Impossible & Limit \\
\midrule
optimistic & exact & $6$ & $48$ & $365$ \\
optimistic & approximating & $6$ & $0$ & $413$ \\
pessimistic & exact & $20$ & $82$ & $317$ \\
pessimistic & approximating & $20$ & $82$ & $317$ \\
probabilistic & exact & $20$ & $70$ & $329$ \\
probabilistic & approximating & $16$ & $81$ & $322$ \\
\end{tabular}
\caption{
  Outcomes.
  All queries are assertions that seem to be violated when the cheapest abstraction is used.
  A \df{ruled out} query is an assertion that is shown not to be violated.
  An \df{impossible} query is an assertion that seems violated
    even if the most precise abstraction is used.
  The exact optimiser is MiFuMax~\cite{mifumax}.
  The approximating optimiser is based on MCSls~\cite{mcsls}.
}
\label{table:outcomes}
\end{table}
\begin{figure}
\hbox to \linewidth{\hss\includegraphics[width=\linewidth]{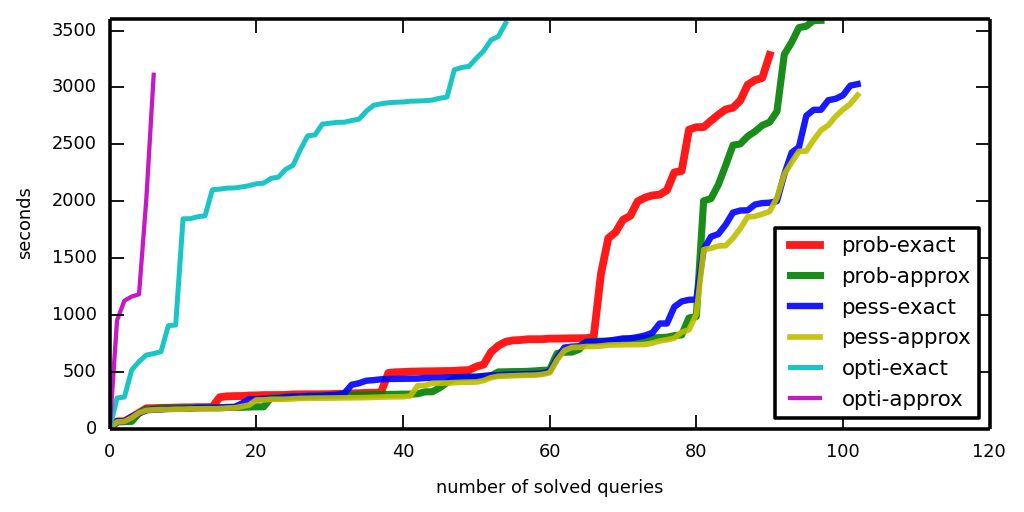}\hss}
\caption{
  Runtime comparison.
}
\label{figure:cactus-all}
\end{figure}
% >>>

We saw a refinement algorithm (\autoref{section:refinement-algorithm})
  whose key step chooses an abstraction to try next.
Then we saw how to estimate whether an abstraction $a'$ is a good choice
  (\autoref{section:refinement-termination} and \autoref{section:refinement-balance}):
  it should have a high ratio between success probability and runtime cost.
But, since the number of abstractions is exponential in the number of parameters,
  it is infeasible to enumerate all in the search for the best one.
Instead of performing a naive exhaustive search,
  we encode the search problem as a $\maxsat$ problem.

Let us summarise the search problem.
Given are a query~$q$, an abstraction~$a$ and its local provenance~$G^a$.
We want to find an abstraction $a'>a$
  that maximises the ratio $\Pr(a' \in {\bf A}_{\rm n}^{\approx})/c(a')$,
  where $c(a')$~is an estimate of the runtime of the analysis under abstraction~$a'$
  (see \autoref{corollary:greedy-choice}).
We will approximate $\Pr(a'\in {\bf A}_{\rm n}^{\approx})$ by a lower bound
  (see \autoref{lemma:noprob-lowerbound}).
Based on empirical observations,
  we estimate the runtime of the analysis to increase exponentially
    with the number $\sum_{x\in P}a(x)$ of precise parameters.
In short, we want to evaluate the following expression:
\begin{align*}
  \argmax_{\substack{a'\\ a'>a}}
  \Biggl(
    \biggl(
      \max_{\substack{H\\ H\subseteq G^a\\q\in\reach{H}{(T(a,a'))}}} \prod_{e\in H} \E {\bf S}_e
    \biggr)
    \bigg/
    \exp\Bigl(\alpha\sum_{x\in P} a'(x)\Bigr)
  \Biggr)
\end{align*}
Or, after absorbing $\max$ in $\argmax$, taking the log of the resulting objective value,
and simplifying the outcome:
\begin{align}\label{eq:maxsat-argmax}
  \argmax_{\substack{
    a', H\\
    a'>a,\ H\subseteq G^a\\
    q\in\reach{H}{(T(a,a'))}
  }}
  \biggl(
    \sum_{e\in H} \log (\E {\bf S}_e) - \sum_{\substack{x\in P\\ a'(x)=1}} \alpha
  \biggr)
\end{align}
We shall evaluate this expression by using a $\maxsat$ solver.
The idea is to encode the range of $\argmax$ as hard constraints,
  and the objective value as soft constraints.

There exist several distinct versions of the $\maxsat$ problem.
We define here a version that is most convenient to our development.
We consider arbitrary boolean formulas, not necessarily in some normal form.
We view assignments as sets of variables; in particular,
\begin{align*}
  M &\models x  &&\text{iff} && x\in M \\
  M &\models \bar{x} &&\text{iff} && x\notin M \\
  M &\models \phi_1 \land \phi_2 &&\text{iff}
    &&\text{$M\models \phi_1$ and $M\models\phi_2$}
\end{align*}
The evaluation rules for other boolean connectives are as expected.
If $M\models \phi$ holds, we say that the assignment~$M$ is a \df{model} of formula~$\phi$.
\begin{problem}[$\maxsat$]\label{problem:maxsat}
Given are a boolean formula~$\Phi$
  and a weight $w(x)$ for each variable~$x$ that occurs in~$\Phi$.
Find a model~$M$ of~$\Phi$ that maximises $\sum_{x\in M} w(x)$.
\end{problem}
We refer to $\Phi$ as the \df{hard constraint}.
\begin{remark}
Technically, \autoref{problem:maxsat} is none of the standard variations of $\maxsat$.
It is easy to see, although we do not prove it here,
  that \autoref{problem:maxsat} is polynomial-time equivalent
  to partial weighted $\maxsat$~\cite{morgado:constraints13,maxsat-survey-1}:
  the reduction in one direction uses the Tseytin transformation,
  while the reduction in the other direction introduces relaxation variables.
\end{remark}

The idea of the encoding is to define the hard constraint~$\Phi$ such that
(i)~the models of~$\Phi$ are in one-to-one correspondence
  with the possible choices of $H$~and~$T$ such that
    $H \subseteq G^a$ and $P_0(a) \subseteq T \subseteq P_0(a)\cup P_1(a)$,
  and moreover
(ii)~each model also encodes the reachable set $\reach{H}{T}$.
To construct a hard constraint~$\Phi$ with these properties,
  we use the same technique as we used for computing the likelihood
  (\autoref{section:learning-likelihood}%
  \longshort{ and \autoref{section:lemma-learning-proof}}{}).
As was the case for likelihood, cycles lead to an exponential explosion.
We again deal with cycles by retaining only forward arcs:
\begin{align*}
  G^a_{\to} \;&\defeq\;
    \{\,e\in G^a\mid\text{$e$ is a forward arc w.r.t. $P_0(a)\cup P_1(a)$}\,\}
\end{align*}
The hard constraint is a formula whose variables correspond to vertices and arcs of~$G^a_{\to}$.
More precisely, its set of variables is $X_V(G^a_{\to})\cup X_E(G^a_{\to})$, where
$$
X_V(G) \defeq \{x_u\mid\text{$u$ vertex of $G$}\}
\quad
X_E(G) \defeq \{x_e\mid\text{$e$ arc of $G$}\}
$$
%\begin{align*}
%  X_V(G) \;&\defeq\; \{\,x_u\mid\text{$u$ vertex of $G$}\,\} \\
%  X_E(G) \;&\defeq\; \{\,x_e\mid\text{$e$ arc of $G$}\,\}
%\end{align*}
We construct the hard constraint $\Phi$ as follows:
\begin{align}
\begin{aligned}
  \Phi \;&\defeq\;
  \mathop{\exists}_{e\in G^a_{\to}}\!\! y_e \;
    \Bigl(\Phi_1 \land \Phi_2 \land \Phi_3\Bigr)
\\
  \Phi_1 \;&\defeq\;
    \hskip-1em
    \bigwedge_{\substack{e=(h,B)\in G^a_{\to}}}
    \hskip-.8em
    \biggl(
        \Bigl(y_e \iff \Bigl(x_e \wedge \bigwedge_{b\in B} x_b \Bigr)\Bigr)
        \wedge
        (y_e \implies x_h)
    \biggr) \\
  \Phi_2 \;&\defeq\;
    \hskip-1em
    \bigwedge_{\substack{h\\\text{vertex of $G^a_{\to}$}\\ h \not\in P_0(a) \cup P_1(a)}}
    \hskip-1em
    \biggl(x_h \implies \Bigl( \bigvee_{e=(h,B)\in G^a_{\to}} y_e \Bigr)
    \biggr) \\
  \Phi_3 \;&\defeq\;
    x_q
    \land \Bigl( \bigwedge_{u \in P_0(a)} x_u \Bigr)
    \land \Bigl( \bigvee_{u \in P_1(a)} x_u \Bigr)
\end{aligned}
  \label{eq:maxsat-phi}
\end{align}
The notation $\mathop{\exists}_{e\in G^a_{\to}} y_e$ stands for several existential quantifiers,
  one for each variable in the set $\{\,y_e\mid e\in G^a_{\to}\,\}$.
  %$\{y_e\}_{e\in G^a_{\to}}$.
Intuitively,
  the constraints $\Phi_1$~and~$\Phi_2$ ensure that the models correspond to reachable sets,
  and the constraint $\Phi_3$ ensures that the query is reachable and that $a'>a$.

The formula~$\Phi$ defined above has several desirable properties:
  its size is linear in the size of the local provenance~$G^a$,
  it is satisfiable, and
  each of its models represents a pair $(a',H)$
    that satisfies the range conditions of~\eqref{eq:maxsat-argmax}.
The satisfiability of~$\Phi$
  is important for the correctness of the refinement algorithm,
  and it follows from how we remove cycles, by retaining forward arcs.
To state these properties more precisely,
  let us denote the range of~\eqref{eq:maxsat-argmax} by~$F(G^a)$ where
\begin{align}
\label{eq:maxsat-F}
  F(G) \defeq \{\,(a',H)\mid
    \text{$a'>a$ and $H\subseteq G$ and $q\in\reach{H}{(T(a,a'))}$}\,\}
\end{align}

\begin{restatable}{lemma}{lemmafeasible}\label{lemma:feasible-subset}
Let $a$ be an abstraction, and let $q$ be a query, for some analysis~${\cal A}$.
Let $F(G)$ and $G^a_{\to}$ be defined as above.
If $a < \top$ and $q\in{\cal A}(a)$, then
$
 (\top,G^a_{\to}) \in F(G^a_{\to}) \subseteq F(G^a)
$.
\end{restatable}

The conditions $a<\top$ and $q \in {\cal A}(a)$ are guaranteed to hold when
  $\textsc{ChooseNextAbstraction}$ is called on line~6 of \autoref{figure:refinement-algorithm}.

\begin{restatable}{lemma}{bijectionlemma}\label{lemma:maxsat-bijection}
Let $a$ be an abstraction, and let $q$ be a query, for some analysis~${\cal A}$.
Let the hard constraint $\Phi$ be defined as in \eqref{eq:maxsat-phi}:
  let the feasible set $F(G^a_{\to})$ be defined as in \eqref{eq:maxsat-F}.
There is a bijection
  between the models~$M$ of~$\Phi$ and the elements~$(a',H)$ of~$F(G^a_{\to})$.
According to this bijection,
\begin{align*}
  M \cap X_E (G^a_{\to}) &= X_E(H) \\
  M \cap X_V (G^a_{\to}) &= X_V\bigl(\reach{H}{(T(a,a'))}\bigr)
\end{align*}
\end{restatable}

The proof of this lemma,
  given in \longshort{\autoref{section:proofs-refinement}}{the submitted supplement},
  relies on techniques very similar to those used to prove \autoref{theorem:learning}.

At this point, we know how to define the hard constraint~$\Phi$,
  so that its models form a subrange of the range of~\eqref{eq:maxsat-argmax}.
It remains to encode the value
  $\sum_{e\in H} \log (\E {\bf S}_e) - \alpha\sum_{x\in P} a'(x)$
  by assigning weights to variables.
This is very easy.
Each arc variable~$x_e$ is assigned the weight $w(x_e)=\log(\E {\bf S}_e)$.
Each vertex variable~$x_u$ corresponding to~$u\in P_0(a)\cup P_1(a)$
  is assigned the weight $w(x_u)=-\alpha$.
All other variables are assigned the weight~$0$.

% >>>
% >>>
% vim:spell spelllang=en_gb:fmr=<<<,>>>:

\section{Empirical Evaluation}\label{section:experiments}

% plan
%   - intro
%   - experimental design
%   - numeric optimisation of likelihood
%   - comparison of models
%   - total time

In the empirical evaluation%~
  \footnote{\url{http://rgrig.appspot.com/static/papers/popl2016experiments.html}}
  we aim to answer three questions:
(a)~Which optimisation algorithm should be used for learning
  (\autoref{section:experiments-optimisation})?
(b)~How well does the probabilistic model predict what the analysis does
  (\autoref{section:experiments-prediction})?
(c)~What is the effect of the new refinement algorithm on the total runtime
  (\autoref{section:experiments-runtime})?

\subsection{Experimental Design}\label{section:experiments-design} % <<<

For experiments, our goal was to improve upon the refinement algorithm of
  Zhang~et~al.~\cite{zhang:pldi14}.
Accordingly, we use the same test suite and the same aliasing analysis.
The test suite consists of 8~Java programs,
  which amount to $0.45\,\rm MiB$ of application bytecode plus $1\,\rm MiB$ of library bytecode.

We try three refinement strategies: optimistic, pessimistic, and probabilistic.
The optimistic strategy uses the baseline refinement algorithm.
The pessimistic strategy uses our refinement algorithm with all hyperparameters set to~$1$.
The probabilistic strategy uses our refinement algorithm with hyperparameters learnt.
We use a time limit of $60$~minutes per query, and a memory limit of $25\,\rm GiB$.

For learning, we observe what the analysis does on a small set of queries and abstractions.
Each observation is essentially an event of the form
  `$\reach{\randomH}{T_1}=R_1$ and $\ldots$ and $\reach{\randomH}{T_n}=R_n$'
  (\autoref{section:learning-experiment}).
From these observations we learn hyperparameters,
  by optimising a lower bound on the likelihood (\autoref{section:learning-likelihood}).
The hyperparameters we use to solve a query
  are learnt only from observations made on the other programs.

% >>>
\subsection{Numeric Optimisation of Likelihood}\label{section:experiments-optimisation} % <<<

First, from the $8$~programs, we chose a random sample of $26$~queries.
Then, for each query, we chose a random sample of~$10$ abstractions
  (\autoref{section:learning-experiment}).
In total, the training set has $260$~samples.

We first tried three numerical optimisers from the SciPy toolkit~\cite{scipy}:
  {\tt tnc}, {\tt slsqp}, and {\tt basinhopping}.
They all fail.
Then we implemented a couple of numeric optimisers ourselves.
We found that the cyclic coordinate ascent method works well on our problem.
In the implementation, we use {\tt basinhopping} and {\tt slsqp} as subroutines,
  for line search.

%We first tried several off the shelf optimisers from the SciPy toolkit~\cite{scipy}:
%  {\tt tnc}, {\tt slsqp}, {\tt basinhopping}.
%Initially, they all failed.
%But, if we weaken the lower bound enough by keeping only one forward arc per vertex
%  (\autoref{section:learning-likelihood}),
%  then {\tt slsqp} obtains good hyperparameters, quickly.
%
%We then implemented two optimisers tailored to our problem: {\tt hill} and {\tt coord}.
%The optimiser {\tt hill} implements gradient ascent with an exponentially decreasing step,
%  and with support for the bounds $[0,1]$.
%The optimiser {\tt coord} implements cyclic coordinate ascent,
%  and uses {\tt basinhopping} for line search.
%The optimiser {\tt basinhopping} uses the Metropolis algorithm,
%  but also improves proposals using a local search algorithm;
%  we used {\tt slsqp} as the local search algorithm.
%The optimiser {\tt slsqp} uses sequential least squares programming.
%
%This time we only weakened the log-likelihood
%  as much as it was necessary for it to fit in memory.
%The optimizer {\tt hill} finds a log-likelihood of $-240$ after half an hour.
%The optimizer {\tt coord} finds a log-likelihood of $-137$ after a few seconds.

%Cyclic coordinate descent works best for our problem.
Intuitively, cyclic coordinate ascent behaves well because
  the likelihood tends to be concave along a coordinate,
  and tends to not be concave along an arbitrary direction.
Concave functions are much easier to optimise than non-concave functions,
  and so the line search algorithm has an easier task when applied along coordinates.

% >>>
\subsection{Predictive Power of the Probabilistic Model}\label{section:experiments-prediction} % <<<

In addition to the $260$~samples used for training,
  we obtain, using the same method, another set of $260$~samples used for evaluation.
Given a model, which is determined by an assignment of values to hyperparameters,
  we can evaluate likelihood bounds for each of the $260$~evaluation samples.
In absolute terms, these numbers are hard to interpret:
  are they good or bad?
To make the numbers more meaningful,
  we consider three models, and we see how good they are relative to each other.

The three models are: {\tt fine}, {\tt coarse}, and {\tt all-one}.
The {\tt fine} model is learnt as described above.
The {\tt coarse} model is also learnt as described above,
  but under the constraint that all hyperparameters have the same value.
The {\tt all-one} model simply assigns value~$1$ to all hyperparameters,
  and thus corresponds to the pessimistic refinement strategy.

\autoref{table:average-likelihood}
  presents the results of the three models on the evaluation set.
For the aliasing analysis we consider,
  it turns out that an abstraction chosen at random does no better than the cheapest abstraction
  in $95\%$~of cases.
The {\tt all-one} model predicts that all abstractions do no better than the cheapest one,
  so it is exactly right in these $95\%$~of cases;
  conversely, it thinks the other $5\%$~of cases cannot happen.
More interestingly,
  the {\tt fine} model thinks that $1.2\%$ samples from the evaluation set cannot happen.
This means that some hyperparameter is~$1$ but should be~$<1$.
We expect that the number of such situations would decrease as the training set grows.
Assuming this is true,
  we can conclude that the {\tt fine} model is better than the {\tt coarse} model.

It is not possible to conclude which of {\tt all-one} and {\tt fine} is better.
One difficulty is that the $95\%$ is a property of the analysis.
It might very well be that for another analysis this percent
  (of cases in which precision helps) is higher or lower.
A lower percentage would favour the {\tt fine} model;
  a high percentage favours the {\tt all-one} model.

% >>>
\subsection{Total Analysis Runtime}\label{section:experiments-runtime}  % <<<

In the $8$~programs there are in total $1450$~queries.
We report results for a random sample of $419$~queries.
The first thing to notice in \autoref{table:outcomes} is that most queries are not solved.
This is in stark contrast with Zhang~et~al.~\cite{zhang:pldi14}
  where all queries are reported as solved.
The difference is explained by several differences between their setup and ours.
(1)~In addition to their PolySite queries, we also include Downcast queries.
The latter are more difficult.
(2)~We used less space and time:
  they used a machine with $128\,\rm GiB$ of memory, whereas we only had $25\,\rm GiB$ available;
  they did not have an explicit time limit, whereas we used $1$~hour as our time limit.
(3)~One of our modifications to the code (unfortunate, with hindsight),
  was that we loaded in memory the results of the Datalog analysis,
  which further increased our memory use.
(4)~They solve multiple queries at once, whereas we solve one at a time.
  By solving one query at a time, we can make a more fine grained comparison.

These differences notwithstanding,
  we stress that the results reported here are for running different algorithms
  under conditions that are as similar as possible.
For example, as much as possible of the implementation is shared.

From the number of solved queries (\autoref{table:outcomes}),
  we see that the refinement strategies, from best to worst, are:
  pessimistic, probabilistic, optimistic.
The pessimistic strategy solves the same set of $102$~queries
  regardless of the optimiser it uses.
The probabilistic strategy solves $101$~queries in total,
  if we take the union over the two optimisers.
There is exactly one query solved by the pessimistic strategy but not by the probabilistic one.
The pessimistic strategy solves this query in four iterations,
  whereas the probabilistic strategy dies in the second iteration.
The exact optimiser times out.
The approximate optimiser increases the precision more than necessary after the first iteration,
  the Datalog solver does cope with the increased precision,
  but an out of memory error happens while Datalog's answer is loaded in memory.

\autoref{figure:cactus-all} compares the six configurations from the point of view of runtime.
We see that both the pessimistic and the probabilistic strategies
  are better than the optimistic strategy.

% >>>
\subsection{Discussion}\label{section:experiments-discussion} % <<<

According to \autoref{table:outcomes} and \autoref{figure:cactus-all},
  setting all hyperparameters to~$1$ works better than using learnt hyperparameters.
Given this, is there any point in learning hyperparameters?
We believe the answer is yes.
Initially we tried only an exact $\maxsat$ solver%
  \footnote{also, at submission time, we had not tried setting all hyperparameters to~$1$}\!.
When the pessimistic strategy succeeds but the probabilistic strategy fails,
  the cause is always that the $\maxsat$ solver times out.
Our encoding in $\maxsat$ is already an approximation, so an approximate answer would do.
We conjectured that replacing the exact solver with an approximate one would improve performance.
We are not aware of an off-the-shelf approximate $\maxsat$ solver,
  so we implemented one.
Comparing {\tt prob-exact} with {\tt prob-approx},
  we see that using an approximate solver does improve the results, but not enough.
However, our approximate solver is so dumb
  that we feel it ought to be possible to do much better.

Another reason to learn hyperparameters is independent of their use for refinement:
  learnt hyperparameters identify interesting parts of an analysis implemented in Datalog
  (\autoref{section:learning-results}).
This is especially useful when one wants to understand an analysis implemented by a third party.

Finally,
  we note that our empirical evaluation of refinement strategies
  shows promise but is not comprehensive.
In future work,
  we intend to try better approximate $\maxsat$ solvers,
  and we intend to evaluate refinement algorithms on more analyses implemented in Datalog.
But, first,
  we need better approximate $\maxsat$ solvers,
  and we need more analyses implemented in Datalog.

% >>>

% vim:spell spelllang=en_gb:fmr=<<<,>>>:

%\overfullrule=10pt

\section{Related and Future Work}\label{sec:related}

The potential of using machine learning techniques or probabilistic
reasoning for addressing challenges in static analysis~\cite{Cousot77,slam:popl02}
has been explored by several researchers in the past ten years.
Three dominant directions so far are: to infer program specifications
automatically using probabilistic models or other inductive learning
techniques~\cite{kremenek:ijcai07,livshits:pldi09,Sankaranarayanan:icse08,beckman:pldi11,Mishne:oopsla12,raychev:pldi14,Raychev-popl15},
to guess candidate program invariants from test data or program traces
using generalisation techniques from machine learning~\cite{sharma:cav12,nori:fse013,loginov:cav05},
and to predict properties of potential or real program errors, such as true positiveness
and cause, probabilistically~\cite{liblit:pldi03,liblit:pldi05,zheng:icml06,yi:ipl07}.
Our work brings a new
dimension to this line of research by suggesting the use of a probabilistic model
for predicting the effectiveness of program abstractions: a probabilistic model can be
designed for predicting how well a parametric static analysis would perform
for a given verification task when it is given a particular abstraction, and
this model can help the analysis to select a good program abstraction for the task
in the context of abstraction refinement. Another important message
of our work is that the derivations computed during each analysis run include a large amount
of useful information, and exploiting this information could lead to more
beneficial interaction between probabilistic reasoning and static analysis.

Machine learning techniques have been used before to speed up abstraction refinement%
  ~\cite{clarke:cav02,gupta2006},
  but in the setting of bounded model checking of hardware.
%The techniques used there (learning decision trees)
%  are different from ours (variational inference).
%Moreover, the machine learning techniques are used with a different goal,
%  to find a small description of a separator.
%Intuitively, in our setting this corresponds to solving the optimisation problem,
%  for which we use MaxSAT solver rather than machine learning techniques.

%A typical bottleneck in combining techniques from probabilistic reasoning
%with techniques from static analysis is that the former are inherently numeric
%while the latter are not. To bridge the gap, one needs to design so called
%features, which essentially translate between the non-numeric world of static analysis
%to the numeric world of probabilities and machine learning. But, designing such features
%is no easy task. Our work shows that it is possible to obtain good results without
%designing any feature at all, provided only that the analysis is implemented in Datalog.

Several probabilistic models for program source code have
been proposed in the past \cite{Hindle-ICSE12,Maddison-ICML14,Allamanis-FSE14,Allamanis-ICML15,raychev:pldi14,Karaivanov-onward14,Raychev-popl15},
and used for extracting natural coding conventions~\cite{Allamanis-FSE14},
helping the correct use of library functions~\cite{raychev:pldi14}, translating programs
between different languages~\cite{Karaivanov-onward14}, and cleaning program source code
and inferring likely properties~\cite{Raychev-popl15}. These models are different from ours
in that they are not designed to predict the behaviours of program analyses
under different program abstractions, the main task of our probabilistic
models.

Our probabilistic models are examples of first-order probabilistic logic
programs studied in the work on statistical relational
learning \cite{Sato-ICLP95,problog-main,srl-book,lfi-problog}.
%Our learning task is an instance of LFI (learning from partial interpretations),
%  for which other algorithms have been proposed~\cite{lfi-problog}.
In our case, models are large, and training data provides only partial information
about the random variable~$\randomH$ used in the models.
To overcome this difficulty, we designed an algorithm tailored to our needs,
  which is based on the idea of variational inference%
  ~\cite{Wainwright-variational08,Jordan-variational99}.
More precisely, we optimised a lower bound on the likelihood.
%Learning hyperparameters in such cases is generally intractable,
%  and we have overcome this intractability
%by analytically deriving the lower bound of probabilities in \autoref{theorem:learning}
%and optimising this lower bound. Using such a proxy during learning is common
%in machine learning, in particular, in the work on variational
%inference~\cite{Wainwright-variational08,Jordan-variational99}.

Our work builds on a large amount of research for automatically finding
good program abstraction, such as CEGAR~\cite{slam:popl02,clarke:cav02,clarke:jacm03,magic,blast:popl04,rybalchenko:vmcai07}, parametric static analysis with parameter search
algorithms~\cite{liang:pldi11,naik:popl12,zhang:pldi13,zhang:pldi14},
and static analysis based on Datalog or Horn solvers~\cite{whaley:aplas05,doop1,doop2,rybal:pldi12,rybal:tacas12}.
The novelty of our work lies in the use of adding a bias in this abstraction search
using a probabilistic model, which predicts the behaviour of the static analysis under
different abstractions.

One future direction would be to find new applications for our probabilistic techniques.
For example, one could try to use our techniques
  in order to improve other, non-probabilistic approaches
  to estimating the impact of abstractions~\cite{smaragdakis:pldi14,oh:pldi14}.
Another future direction would be
  to better characterise the theoretical properties of our refinement algorithm.
For example, if applied in the setting of abstract interpretation,
  how does it interact with the notion of completeness%
  ~\cite{giacobazzi-metaanalysis,giacobazzi-complete}?

%Recently, non-probabilistic approaches for estimating the impacts of
%different program abstractions on a given analysis or verification task have been
%proposed~\cite{smaragdakis:pldi14,oh:pldi14}.
%One interesting future direction is to revisit these approaches from the perspective of probabilistic modelling explained in this paper,
%with the goal of obtaining probabilistic variants of their approaches that
%replace the current hard-coded heuristics for prediction by adaptable ones.

% vim:spell spelllang=en_gb:

\section{Conclusion}\label{sec:conclusions} % <<<

We have presented a new approach to abstraction refinement,
  one that receives guidance from a learnt probabilistic model.
The model is designed to predict how well would the static analysis
  perform for a given verification task under different parameter settings.
The model is fully derived from the specification of the analysis,
  and does not require manually crafted features.
Instead,
  our model's prediction
    is based on all the reasoning steps performed by the analysis in a failed run.
%  which are represented as a hypergraph in the paper.
To make these predictions,
  the model needs to know how much approximation is involved in each Datalog rule
  that implements the static analysis.
We have shown how to quantify the approximation,
  by using a learning algorithm that observes the analysis running on a large codebase.
%The result of the learning algorithm are the so called hyperparameters,
%  which quantify the amount of approximation in each Datalog rule,
%  and which enable our model to make its predictions.
Finally, we have shown how to combine the predictions of the model with a cost measure
  in order to choose an optimal next abstraction to try during refinement.
Our empirical evaluation with an object-sensitive pointer analysis
  shows that our approach is promising.
%  shows the promise of our approach.

%We have presented a new approach for
%abstraction refinement that receives guidance from
%a learnt probabilistic model. Our model is designed to predict how well a parametric
%static analysis would perform for a given
%verification task under different parameter settings.
%It is fully derived from the specification of the analysis, and 
%does not require manually-crafted features. Instead,
%our model's prediction is based on all the reasoning steps performed by 
%the analysis in a failed run, which
%are represented as a hypergraph in the paper. We have described how hyperparameters 
%of the model can be learnt from prior runs of the same analysis, and have presented 
%a method for combining the prediction of the model
%and the cost measure and choosing an optimal  next abstraction to try during
%refinement. Our empirical evaluation with an object-sensitive pointer analysis
%shows the promise of our approach.

% >>>
% vim:spell spelllang=en_gb:fmr=<<<,>>>:

\acks
We thank Mikol\'a\v{s} Janota and Xin Zhang for suggesting that we invert Datalog implications.
We thank Mayur Naik for giving us access to the private parts of Chord~\cite{jchord}.
We thank Yongsu Park for giving us access to a server on which we ran preliminary experiments.
We thank reviewers for their suggestions on how to improve the paper.
This work was supported by EPSRC Programme Grant Resource Reasoning (EP/H008373/2).
Yang was also supported by EPSRC and Institute for Information \& communications Technology Promotion (IITP)
grant funded by the Korea government (MSIP)
(No. R0190-15-2011, Development of Vulnerability Discovery Technologies for IoT Software Security).

{\softraggedright

\balance
\bibliography{db}
\bibliographystyle{plain}

}
\longshort{
  \newpage\appendix

\section{Proof of
  \texorpdfstring%
    {\hyperref[theorem:learning]{Theorem~\ref*{theorem:learning}}}
    {Theorem~\ref*{theorem:learning}}
}
\label{section:lemma-learning-proof}
{\allowdisplaybreaks[1]

We begin by restating in our notation a standard result from logic programming.
A \df{dependency graph} of a hypergraph~$G$ is a directed graph
  that includes an arc $(h,b)$ whenever $(h,B)\in G$ and $b\in B$ for some~$B$.
A \df{loop}~$L$ of a hypergraph~$G$ is a nonempty subset of its vertices
  that induce a strongly connected subgraph of the dependency graph of~$G$.
Note that loops are not required to be maximal.
In particular, sets that contain single vertices are loops, called \df{trivial loops}.
The set~$J_G(L)$ of \df{justifications} for loop~$L$ in~$G$ is defined as follows:
\begin{align*}
  J_G(L) \defeq \{\,(h,B)\in G\mid\text{\rm $h\in L$ and $B\cap L=\emptyset$}\,\}
\end{align*}
For a hypergraph~$G$ we define
  its \df{forward formula}~$\phi_{\rightarrow}(G)$
  and its \df{backward formula}~$\phi_{\leftarrow}(G)$
  as follows:
\begin{align*}
  \phi_{\rightarrow}(G) \;&\defeq\;
    \bigwedge_{e=(h,B)\in G}
      \biggl(
      \Bigl( \bigl( \bigwedge_{b\in B} x_b \bigr) \iff x_e \Bigr)
      \land (x_e \implies x_h)
      \biggr)
\\
  \phi_{\leftarrow}(G) \;&\defeq\;
    \bigwedge_{\substack{L\\\text{\rm loop of~$G$}}}
      \biggl(
        \bigl( \bigwedge_{u\in L} x_u \bigr) \implies \bigl( \bigvee_{e\in J_G(L)} x_e\bigr)
      \biggr)
\end{align*}
Both formulas are defined over the following set of variables:
\begin{align*}
  \{\,x_u\mid\text{\rm $u$ vertex of $G$}\,\}
  \cup
  \{\,x_e\mid\text{\rm $e$ arc of $G$}\,\}
\end{align*}
We define the \df{formula}~$\phi(G)$ of a hypergraph~$G$ by
\begin{align*}
  \phi(G) \;\defeq\;
    \mathop{\exists}_{e\in G}\!\! x_e \;
    \bigl(\phi_{\rightarrow}(G) \land \phi_{\leftarrow}(G)\bigr)
\end{align*}
The notation $\mathop{\exists}_{e\in G} x_e$ stands for several existential quantifiers,
  one for each variable in the set $\{x_e\}_{e\in G}$ indexed by~$G$.
In the definition of $\phi(G)$ from above,
  the existential quantification is not strictly necessary, but convenient:
Because the remaining free variables correspond to vertices,
  sets of variables are isomorphic to sets of~vertices.

We view models~$M$ of a formula~$\varphi$ as sets of variables; that is,
\begin{align*}
  M &\models x  &&\text{iff} && x\in M \\
  M &\models \bar{x} &&\text{iff} && x\notin M \\
  M &\models \varphi_1 \implies \varphi_2 &&\text{iff}
    &&\text{$M\models \varphi_1$ implies $M\models\varphi_2$} \\
  M &\models \exists x\,\varphi &&\text{iff}
    &&\text{$M\models\varphi[x:=0]$ or $M\models\varphi[x:=1]$}
\end{align*}
and so on, in the standard way.
There is an obvious one-to-one correspondence between sets of vertices and models;
  if $S$~is a set of vertices, we write $XS$ for the corresponding model,
  which is a set of variables:
\begin{align*}
  XS \;\defeq\; \{\,x_s\mid s\in S\,\}
\end{align*}
The following result is stated in~\cite[Section~3]{lee:loops},
  in a slightly more general form and with slightly different notations:

\begin{lemma}\label{lemma:loop-formulas-standard}
Let $G$ be a hypergraph, and let $\phi(G)$ be its formula, defined as above.
Then $X\bigl(\reach{G}{\emptyset}\bigr)$ is the unique model of~$\phi(G)$.
\end{lemma}

For the proof, we refer to~\cite{lee:loops}.

\begin{remark}
We note that $\phi_{\rightarrow}(G)$~is linear in the size of~$G$,
  while $\phi_{\leftarrow}(G)$~is exponential in the size of~$G$ in the worst case.
One could wonder whether it is possible to define~$\phi(G)$
  in a way that does not lead to exponentially large formulas
  but \autoref{lemma:loop-formulas-standard} still holds.
It turns out there are reasons to suspect
  that such an alternative definition does not exist~\cite{razborov:many-loops}.
\end{remark}

Here, we shall need a more flexible form of \autoref{lemma:loop-formulas-standard}.
Let $S$ be a distinguished subset of vertices, none of which occurs in the head of an arc.
Define
\begin{align*}
  \phi^S_{\leftarrow}(G) \;&\defeq\;
    \bigwedge_{\substack{L\\\text{\rm loop of~$G$}\\L\cap S=\emptyset}}
      \biggl(
        \bigl( \bigwedge_{u\in L} x_u \bigr) \implies \bigl( \bigvee_{e\in J_G(L)} x_e\bigr)
      \biggr)
\end{align*}
and
\begin{align}
 \phi^S(G)\;\defeq\;\mathop{\exists}_{e\in G} x_e\,
  \bigl(\phi_{\rightarrow}(G) \land \phi^S_{\leftarrow}(G) \bigr)
  \label{eq:graph-formula-relative}
\end{align}

\begin{corollary}\label{corollary:loop-formulas-flexible}
Let $G$ be a hypergraph,
  let $S$ be a subset of vertices such that none of them occurs in the head of an arc,
  and let $\phi^S(G)$ be defined as above.
For each subset~$T$ of~$S$,
  there exists a unique model~$M$ of~$\phi^S(G)$ such that $X^{-1}(M)\cap S=T$,
  namely $M=X\bigl(\reach{G}{T}\bigr)$.
\end{corollary}

\begin{proof}
For a fixed but arbitrary $T\subseteq S$, construct the graph
\begin{align*}
  G_T \;\defeq\; G \cup \{\,(t,\emptyset)\mid t\in T\,\}
\end{align*}
It is easy to check that $\reach{G}{T}=\reach{G_T}{\emptyset}$.
From \autoref{lemma:loop-formulas-standard},
  we know that $X\bigl(\reach{G_T}{\emptyset}\bigr)$ is the unique model of~$\phi(G_T)$.
Since the vertices of~$S$ do not occur in the heads of arcs, they appear only in trivial loops.
Thus, we have
\begin{align*}
  \phi_{\rightarrow}(G_T) \;&=\;
    \phi_{\rightarrow}(G) \land \Bigl( \bigwedge_{t\in T} x_t \Bigr)
\\
  \phi_{\leftarrow}(G_T) \;&=\;
    \phi^S_{\leftarrow}(G) \land \Bigl( \bigwedge_{s\in S\setminus T} \bar{x_s} \Bigr)
\end{align*}
(The formulas above eliminate via existential quantification
  the variables corresponding to the dummy arcs $(t,\emptyset)$ of~$G_T$,
  but this is of little consequence.)
And finally
\begin{align*}
\begin{aligned}
  \phi_{\rightarrow}(G_T) \land \phi_{\leftarrow}(G_T)
  \;&=\;
  \phi_{\rightarrow}(G) \land \phi^S_{\leftarrow}(G) \\
  &\qquad  \land \Bigl( \bigwedge_{s\in S\setminus T} \bar{x_s} \Bigr)
    \land \Bigl( \bigwedge_{t\in T} x_t \Bigr)
\end{aligned}
\end{align*}
This concludes the proof.
\end{proof}

We now take a special case of \autoref{corollary:loop-formulas-flexible}.

\begin{corollary}\label{corollary:loop-formulas-sources}
Let $G$ be a hypergraph.
Let $(S,V)$ be a partition of its vertices
  such that no vertex in~$S$ occurs as the head of an arc.
Let $\phi^S(G)$ be defined as above.
Let $R$ be a subset of~$V$.
Define
\begin{align*}
  \phi^{S,R}(G) \;\defeq\;
  \mathop{\exists}_{u\in V}\!\! x_u\;
    \biggl(
      \phi^S(G)
      \land \Bigl(\bigwedge_{u\in R} x_u\Bigr)
      \land \Bigl(\bigwedge_{u\in V\setminus R} \bar{x_u}\Bigr)
    \biggr)
\end{align*}
For all $T\subseteq S$,
  we have that $XT$ is a model of~$\phi^{S,R}(G)$ if and only if $\reach{G}{T}=T\cup R$.
\end{corollary}
\begin{proof}
Let $T$ be a subset of $S$. Then, $XT$ is a model of~$\phi^{S,R}(G)$ 
if and only if $X(T \cup R)$ is a model of $\phi^S(G)$. But
by \autoref{corollary:loop-formulas-flexible}, this is equivalent to
$\reach{G}{T}=T\cup R$.
\end{proof}

The key idea of our proof is to use \autoref{corollary:loop-formulas-sources}
  in such a way that subsets of~$S$ correspond to predictive provenances~$H$.
To this end,
  we define the \df{extended cheap provenance}~$G^{\bot}_T$
  with respect to the set~$T$ of vertices by
\begin{align*}
  G^{\bot}_T \defeq
    \{\,(h,B\cup\{s_e\})\mid e=(h,B)\in G^{\bot}\,\}
    \cup
    \{\,(t,\emptyset)\mid t\in T\,\}
\end{align*}
Recall our notation $G^{\bot}$ for the cheap provenance.
For a predictive provenance $H\subseteq G^{\bot}$, let us write $SH$ for $\{\,s_e\mid e\in H\,\}$.
All the vertices of $S G^{\bot}$ are fresh: they appear in $G^{\bot}_T$ but not in~$G^{\bot}$.
The extended cheap provenance has the property that
\begin{align}
  \reach{G^{\bot}_T}(SH) = (SH)\cup\reach{H}{T}
  \label{eq:Gtt-property}
\end{align}
for all predictive provenances $H\subseteq G^{\bot}$ and all sets of vertices~$T$.

Suppose the cheap provenance~$G^{\bot}$ and two subsets $T$~and~$R$ of its vertices are given.
The following lemma shows how to construct a boolean formula whose models
  are in one-to-one correspondence with the cheap provenances $H\subseteq G^{\bot}$
  for which $R=\reach{H}{T}$.

\begin{lemma}\label{lemma:predictive-provenance-as-model}
Let $G^{\bot}$ be a cheap provenance,
  and let $R$~and~$T$ be two subsets of its vertices.
Define the extended cheap provenance $G^{\bot}_T$ with respect to~$T$ as above.
We have that $R=\reach{H}{T}$ if and only if
  $X(SH)$~is a model of~$\phi^{SG^{\bot},R}(G^{\bot}_T)$.
\end{lemma}

\begin{proof}
In \autoref{corollary:loop-formulas-sources},
  set $S \defeq S G^{\bot}$ and $T \defeq SH$ and $G \defeq G^{\bot}_T$.
We obtain that
\begin{align*}
X(SH) &\models \phi^{S G^{\bot},R}(G^{\bot}_T)
  &&\text{iff}
& \reach{G^{\bot}_T}{(SH)} = (SH) \cup R
\end{align*}
Combining this with \eqref{eq:Gtt-property} we obtain
\begin{align*}
X(SH) &\models \phi^{S G^{\bot},R}(G^{\bot}_T)
  &&\text{iff}
& (SH)\cup\reach{H}{T} = (SH) \cup R
\end{align*}
Finally, since all the vertices in $SH$ are fresh, we are done.
\end{proof}

What remains to be done is to make explicit the formula $\phi^{SG^{\bot},R}(G^{\bot}_T)$
  mentioned in \autoref{lemma:predictive-provenance-as-model}.
This is only a matter of calculation.
We begin by unfolding the definition of $\phi^{SG^{\bot},R}(G^{\bot}_T)$,
  and then that of~$\phi^{SG^{\bot}}(G^{\bot}_T)$.
Below, the notation $\varphi[x_R:=v]$ means that in $\varphi$
  we substitute the variable $x_u$ with value~$v$ for all indices $u\in R$.
Also, we write $V$ for the vertex set of $G^\bot$.
  
\begin{align*}
&\phi^{SG^{\bot},R}(G^{\bot}_T) 
\\&\quad=
  \mathop{\exists}_{u\in V}\!\! x_u\;
    \biggl(
      \phi^{SG^{\bot}}(G^{\bot}_T)
      \land \Bigl(\bigwedge_{u\in R} x_u\Bigr)
      \land \Bigl(\bigwedge_{u\in V\setminus R} \bar{x_u}\Bigr)
    \biggr)
\\&\quad=
  \phi^{SG^{\bot}}(G^{\bot}_T)
    [x_R:=1][x_{V\setminus R}:=0]
\\&\quad=
  \mathop{\exists}_{e\in G^{\bot}_T}\!\! x_e\,
    \bigl(\phi_{\rightarrow}(G^{\bot}_T) \land \phi^{SG^{\bot}}_{\leftarrow}(G^{\bot}_T) \bigr)
    [x_R:=1][x_{V\setminus R}:=0]
\\&\quad=
  \mathop{\exists}_{e\in G^{\bot}_T}\!\! x_e\,
    \bigl(\Psi_{\rightarrow} \land \Psi_{\leftarrow}\bigr)
\end{align*}
where
\begin{align*}
  \Psi_{\rightarrow} \;&\defeq\;
    \phi_{\rightarrow}(G^{\bot}_T)
    [x_R:=1][x_{V\setminus R}:=0]
\\
  \Psi_{\leftarrow} \;&\defeq\;
    \phi^{SG^{\bot}}_{\leftarrow}(G^{\bot}_T)
    [x_R:=1][x_{V\setminus R}:=0]
\end{align*}
Now we calculate $\Psi_{\rightarrow}$~and~$\Psi_{\leftarrow}$, in turn.
We begin with~$\Psi_{\rightarrow}$.
First we unfold the definition of $\phi_{\rightarrow}(G^{\bot}_T)$,
  then we unfold the definition of~$G^{\bot}_T$,
  and finally we apply the substitutions.
During the calculation, we identify $x_{s_e}$ with ${\bf S}_e$.
This is partly notational convenience (to avoid double subscripts),
  but it will also allow us to weigh models according to the probabilistic model.
\begin{align*}
\Psi_{\rightarrow} &=
  \phi_{\rightarrow}(G^{\bot}_T)
  [x_R:=1][x_{V\setminus R}:=0]
\\&=
  \hskip-0.5em
  \bigwedge_{\substack{e\in G^{\bot}_T\\ e=(h,B)}}
  \hskip-0.5em
    \biggl(
    \Bigl( \bigl( \bigwedge_{b\in B} x_b \bigr) \iff x_e \Bigr)
    \land (x_e \implies x_h)
    \biggr)
  \genfrac{[}{]}{0pt}{}{x_R:=1}{x_{V\setminus R}:=0}
\\&=
  \biggl(
  \hskip-1em
  \bigwedge_{\substack{e'\in G^{\bot}\\ e'=(h,B)\\ e=(h,B\cup\{s_{e'}\})}}
  \hskip-0.5em
    \biggl(
    \Bigl( \bigl( \bigwedge_{b\in B\cup\{s_{e'}\}} x_b \bigr) \iff x_e \Bigr)
    \land (x_e \implies x_h)
    \biggr)
\\&\qquad
  {} \land
  \bigwedge_{\substack{t\in T\\ e=(t,\emptyset)}}
  \hskip-0.5em
    (x_e \land x_t)
    \biggr)
  \genfrac{[}{]}{0pt}{}{x_R:=1}{x_{V\setminus R}:=0}
\\&=
  \hskip-1.5em
  \bigwedge_{\substack{e'\in G^{\bot}\\ e'=(h,B)\\ e=(h,B\cup\{s_{e'}\})}}
  \hskip-1.8em
    \biggl(
    \Bigl( \bigl(S_{e'}\land[B\subseteq R]\bigr) \iff x_e \Bigr)
    \land (x_e \implies [h\in R])
    \biggr)
\\&\qquad
  {} \land
  \bigwedge_{\substack{t\in T\\ e=(t,\emptyset)}}
  \hskip-0.5em
    (x_e \land [t\in R])
\end{align*}
If $T\not\subseteq R$, then $\Psi_{\rightarrow}=0$;
otherwise,
\begin{align}
\begin{aligned}
\Psi_{\rightarrow} =
  &\biggl(
  \hskip-2em
  \bigwedge_{\substack{%
    e'=(h,B)\in G^{\bot}\\
    e=(h,B\cup\{s_{e'}\})\\
    \text{\rm $B\subseteq R$ and $h\in R$}
  }}
    \hskip-1.5em
    \bigl(S_{e'}\iff x_e \bigr)
  \biggr)
\land
  \biggl(
  \hskip-2em
  \bigwedge_{\substack{%
    e'=(h,B)\in G^{\bot}\\
    \text{\rm $B\subseteq R$ and $h\not\in R$}
  }}
    \hskip-1.5em
    \bar{S_{e'}}
  \biggr)
\\
  &{} \land
  \biggl(
  \hskip-2em
  \bigwedge_{\substack{%
    e'=(h,B)\in G^{\bot}\\
    e=(h,B\cup\{s_{e'}\})\\
    \text{\rm $B\not\subseteq R$}
  }}
    \hskip-1.2em
    \bar{x_e} \biggr)
  \land
  \biggl(
  \bigwedge_{\substack{%
    t\in T\\
    e=(t,\emptyset)
  }}
    x_e
  \biggr)
\end{aligned}
  \label{eq:psi->}
\end{align}
Next, we calculate $\Psi_{\leftarrow}$.
\begin{align*}
\Psi_{\leftarrow} &=
  \phi^{SG^{\bot}}_{\leftarrow}(G^{\bot}_T)
  [x_R:=1][x_{V\setminus R}:=0]
\\&=
  \bigwedge_{\substack{L\\\text{\rm loop of~$G^{\bot}_T$}\\L\cap SG^{\bot}=\emptyset}}
    \biggl(
      \bigl( \bigwedge_{u\in L} x_u \bigr)
      \implies
      \bigl( \bigvee_{e\in J_{G^{\bot}_T}(L)} x_e\bigr)
    \biggr)
  \genfrac{[}{]}{0pt}{}{x_R:=1}{x_{V\setminus R}:=0}
\\&=
  \bigwedge_{\substack{L\\\text{\rm loop of~$G^{\bot}$}}}
    \biggl(
      \bigl( \bigwedge_{u\in L} x_u \bigr)
      \implies
      \bigl( \bigvee_{e\in J_{G^{\bot}_T}(L)} x_e\bigr)
    \biggr)
  \genfrac{[}{]}{0pt}{}{x_R:=1}{x_{V\setminus R}:=0}
\\&=
  \bigwedge_{\substack{L\\\text{\rm loop of~$G^{\bot}$}}}
    \biggl(
      [L\subseteq R]
      \implies
      \bigl( \bigvee_{e\in J_{G^{\bot}_T}(L)} x_e\bigr)
    \biggr)
\\&=
  \bigwedge_{\substack{L\\\text{\rm loop of~$G^{\bot}$}\\L\subseteq R}}
  \biggl(
    \Bigl(
      \hskip-1em
      \bigvee_{\substack{
        e'\\
        e'=(h,B)\in J_{G^{\bot}}(L)\\
        e=(h,B\cup\{s_{e'}\})
      }}
      \hskip-2em
      x_e
    \Bigr)
    \lor
    \Bigl(
      \bigvee_{\substack{
        t\in T\cap L\\
        e=(t,\emptyset)
      }}
      x_e
    \Bigr)
  \biggr)
\end{align*}
When we calculate $\Psi_{\rightarrow}\land\Psi_{\leftarrow}$ we see that
  $\Psi_{\rightarrow}$ fixes the values of all the variables $x_e$ corresponding to arcs.
\begin{align*}
\hskip2em&\hskip-2em
\phi^{SG^{\bot},R}(G^{\bot}_T)
\;=\;
  \mathop{\exists}_{e\in G^{\bot}_T}\! x_e\, (\Psi_{\rightarrow} \land \Psi_{\leftarrow})
\\&=\;
  [T\subseteq R]
  \land
  \biggl(
  \hskip-2em
  \bigwedge_{\substack{%
    e'=(h,B)\in G^{\bot}\\
    \text{\rm $B\subseteq R$ and $h\not\in R$}
  }}
    \hskip-1.5em
    \bar{S_{e'}}
  \biggr)
%\\&\qquad
  \land
  \Psi_{\leftarrow}
  \left[\substack{\
    \text{\rm $x_e:=S_{e'}$ for $(e,e')\in S$}\\
    \text{\rm $x_e:=0$ for $e\in O$}\\
    \text{\rm $x_e:=1$ for $e\in I$}
  }
  \right]
\end{align*}
where $S$, $O$, and $I$ stand for corresponding ranges in~\eqref{eq:psi->}.
More precisely,
  letting $e'=(h,B)$ range over~$G^{\bot}$
  and letting $e$ be its corresponding arc $(h,B\cup\{s_{e'}\})$ in~$G^{\bot}_T$,
  we have
  $S \;\defeq\;\{\,(e,e')\mid\text{\rm $B\subseteq R$ and $h\in R$}\}$
  and
  $O \;\defeq\;\{\,e\mid B\not\subseteq R\,\}$.
Also, $I$~contains all the dummy arcs of the form $(t,\emptyset)$, for all $t\in T$.
Now we apply these three substitutions to~$\Psi_{\leftarrow}$, one by one.
The first line just introduces a shorthand notation for each of the three kinds of substitutions.
\begin{align*}
\hskip2em&\hskip-2em
  \Psi_{\leftarrow}
  \left[
  \begin{aligned}
  &\text{\rm $x_e:=S_{e'}$ for $(e,e')\in S$}\\
  &\text{\rm $x_e:=0$ for $e\in O$}\\
  &\text{\rm $x_e:=1$ for $e\in I$}
  \end{aligned}
  \right]
=
  \Psi_{\leftarrow}
  \left[
  \begin{aligned}
  &{\cal S} \\ &{\cal O} \\ &{\cal I}
  \end{aligned}
  \right]
\\&=\;
  \bigwedge_{\substack{L\\\text{\rm loop of~$G^{\bot}$}\\L\subseteq R}}
  \biggl(
    \Bigl(
      \hskip-1em
      \bigvee_{\substack{
        e'\\
        e'=(h,B)\in J_{G^{\bot}}(L)\\
        e=(h,B\cup\{s_{e'}\})
      }}
      \hskip-2em
      x_e
    \Bigr)
    \lor
    \Bigl(
      \bigvee_{\substack{
        t\in T\cap L\\
        e=(t,\emptyset)
      }}
      x_e
    \Bigr)
  \biggr)
  \left[
  \begin{aligned}
  &{\cal S} \\ &{\cal O} \\ &{\cal I}
  \end{aligned}
  \right]
\\&=\;
  \bigwedge_{\substack{L\\\text{\rm loop of~$G^{\bot}$}\\L\subseteq R\setminus T}}
  \bigvee_{\substack{
    e'\\
    e'=(h,B)\in J_{G^{\bot}}(L)\\
    e=(h,B\cup\{s_{e'}\})
  }}
  \hskip-2em
    x_e
  \left[
  \begin{aligned}
  &{\cal S} \\ &{\cal O}
  \end{aligned}
  \right]
\\&=\;
  \bigwedge_{\substack{L\\\text{\rm loop of~$G^{\bot}$}\\L\subseteq R\setminus T}}
  \bigvee_{\substack{
    e'\\
    e'=(h,B)\in J_{G^{\bot}}(L)\\
    e=(h,B\cup\{s_{e'}\})\\
    B\subseteq R
  }}
  \hskip-2em
    x_e
  [{\cal S}]
\\&=\;
  \bigwedge_{\substack{L\\\text{\rm loop of~$G^{\bot}$}\\L\subseteq R\setminus T}}
  \bigvee_{\substack{
    e'\\
    e'=(h,B)\in J_{G^{\bot}}(L)\\
    B\subseteq R
  }}
  \hskip-2em
    S_{e'}
\end{align*}

Finally, we conclude that
\begin{align}
\begin{aligned}
  &\phi^{SG^{\bot},R}(G^{\bot}_T) \;=\;
  \\&\qquad
    [T\subseteq R]
    \land
    \biggl(
      \hskip-1em
      \bigwedge_{\substack{e=(h,B)\in G^{\bot}\\B\subseteq R,\  h\notin R}}
      \hskip-1em
        \bar{\bf S}_e
    \biggr)
    \land
    \biggl(
      \bigwedge_{\substack{L\\\text{\rm loop in $G^{\bot}$}\\L\subseteq R\setminus T}}
      \bigvee_{\substack{e=(h,B)\\ e\in J_{G^{\bot}}(L)\\ B\subseteq R}}
        {\bf S}_e
    \biggr)
\end{aligned}
  \label{eq:formula-one-obs}
\end{align}

Now observe that
\begin{align}
\Pr\biggl(\bigcap_{k=1}^n \Bigl( R_k=\reach{\randomH}(T_k)\Bigr) \biggr)
  \;=\;
  \E \biggl(\bigwedge_{k=1}^n \phi^{SG^{\bot},R_k}(G^{\bot}_{T_k}) \biggr)
  \label{eq:conjoining-obs}
\end{align}

Putting together \eqref{eq:formula-one-obs}~and~\eqref{eq:conjoining-obs},
  we obtain the following lemma.

\begin{lemma}\label{lemma:likelihood}
Consider the probabilistic model associated with the cheap provenance~$G^{\bot}$
  of an analysis~${\cal A}$.
Let $T_1,\ldots,T_n$ and $R_1,\ldots,R_n$ be subsets of the vertices of~$G^{\bot}$.
If $T_k\subseteq R_k$ for all~$k$, then
\begin{align*}
\begin{aligned}
&\Pr\biggl( \bigcap_{k=1}^n \bigl( R_k = \reach{\randomH}{T_k} \bigr) \biggr)
  =
\\&\qquad
  \prod_{e\in N} \E \bar{\bf S}_e
  \cdot
  \E \biggl(
  \bigwedge_{\substack{L\\\text{\rm loop of~$G^{\bot}$}}}
  \bigwedge_{\substack{k\\ L\subseteq R_k\setminus T_k}}
  \bigvee_{\substack{e=(h,B)\\ e\in J_{G^{\bot}}(L)\setminus N\\ B\subseteq R_k}}
    {\bf S}_e
  \biggr)
\end{aligned}
\end{align*}
where
\begin{align*}
N \;\defeq\;
  \{\,(h,B)\in G^{\bot}\mid\text{\rm $B\subseteq R_k$ and $h\notin R_k$ for some~$k$}\,\}
\end{align*}
\end{lemma}

\begin{proof}
We assume that $T_k\subseteq R_k$. 
Using \eqref{eq:formula-one-obs}~and~\eqref{eq:conjoining-obs},
we transform $\bigwedge_{k=1}^n \phi^{SG^{\bot},R_k}(G^{\bot}_{T_k})$ as follows:
\begin{align*}
\hskip2em&\hskip-2em
  \bigwedge_{k=1}^n \phi^{SG^{\bot},R_k}(G^{\bot}_{T_k})
\\&=\;
  \bigwedge_{k=1}^n
  \Biggl(
    \biggl(
      \hskip-0.5em
      \bigwedge_{\substack{e=(h,B)\in G^{\bot}\\B\subseteq R_k,\  h\notin R_k}}
      \hskip-1em
        \bar{\bf S}_e
    \biggr)
    \land
    \biggl(
      \bigwedge_{\substack{L\\\text{\rm loop in $G^{\bot}$}\\L\subseteq R_k\setminus T_k}}
      \bigvee_{\substack{e=(h,B)\\ e\in J_{G^{\bot}}(L)\\ B\subseteq R_k}}
        {\bf S}_e
    \biggr)
  \Biggr)
\\&=\;
  \biggl(
    \bigwedge_{e\in N}
      \bar{\bf S}_e
  \biggr)
  \land
  \biggl(
    \bigwedge_{k=1}^n
    \bigwedge_{\substack{L\\\text{\rm loop in $G^{\bot}$}\\L\subseteq R_k\setminus T_k}}
    \bigvee_{\substack{e=(h,B)\\ e\in J_{G^{\bot}}(L)\\ B\subseteq R_k}}
      {\bf S}_e
  \biggr)
\\&=\;
  \biggl(
    \bigwedge_{e\in N}
      \bar{\bf S}_e
  \biggr)
  \land
  \biggl(
    \bigwedge_{k=1}^n
    \bigwedge_{\substack{L\\\text{\rm loop in $G^{\bot}$}\\L\subseteq R_k\setminus T_k}}
    \bigvee_{\substack{e=(h,B)\\ e\in J_{G^{\bot}}(L)\setminus N\\ B\subseteq R_k}}
      {\bf S}_e
  \biggr)
\\&=\;
  \biggl(
    \bigwedge_{e\in N}
      \bar{\bf S}_e
  \biggr)
  \land
  \biggl(
    \bigwedge_{\substack{L\\\text{\rm loop in $G^{\bot}$}}}
    \bigwedge_{\substack{k\in\{1,\ldots,n\}\\ L\subseteq R_k\setminus T_k}}
    \bigvee_{\substack{e=(h,B)\\ e\in J_{G^{\bot}}(L)\setminus N\\ B\subseteq R_k}}
      {\bf S}_e
  \biggr)
\end{align*}
The conclusion of the lemma now follows from the result of this calculation
and the fact that ${\bf S}_e$ and ${\bf S}_{e'}$ are independent whenever $e \neq e'$.
\end{proof}

We can finally prove \autoref{theorem:learning}.
Recall its statement:

\learningtheorem*

If $T_k \not\subseteq R_k$ for some~$k$,
  then the probability and both of its bounds are all~$0$.
In what follows,
  we shall invoke \autoref{lemma:likelihood},
  thus silently assuming that $T_k\subseteq R_k$ for all~$k$.
We first prove the claim about an upper bound, and then show the claim about a lower bound.

\begin{proof}[Proof of the Upper Bound in \autoref{theorem:learning}]
We start with a short calculation which shows what happens if we consider only trivial loops.
Recall the assumption that $h\notin B$ for all arcs $(h,B)$.
\begin{align}
\begin{aligned}
& {}\E \biggl(
  \bigwedge_{\substack{L\\\text{\rm loop of~$G^{\bot}$}}}
  \bigwedge_{\substack{k\\ L\subseteq R_k\setminus T_k}}
  \bigvee_{\substack{e=(h,B)\\ e\in J_{G^{\bot}}(L)\setminus N\\ B\subseteq R_k}}
    {\bf S}_e
  \biggr)
\\&\qquad\qquad {} \;\le\;
  \E \biggl(
  \bigwedge_{\substack{h\\\text{\rm vertex of~$G^{\bot}$}}}
  \bigwedge_{\substack{k\\ h\in R_k\setminus T_k}}
  \bigvee_{\substack{e=(h,B)\\ e\notin N, B\subseteq R_k}}
    {\bf S}_e
  \biggr)
\\&\qquad\qquad {} \;=\;
  \E \biggl(
  \bigwedge_{\substack{h\\C_h \neq \emptyset}}
  \bigwedge_{\substack{k \in C_h}}
  \bigvee_{\substack{e=(h,B)\\ e\notin N, B\subseteq R_k}}
    {\bf S}_e
  \biggr)
\\&\qquad\qquad {} \;=\;
  \prod_{\substack{h\\ C_h \neq \emptyset}}
  \E \biggl(
  \bigwedge_{\substack{k \in C_h}}
  \bigvee_{\substack{e=(h,B)\\ e\notin N, B\subseteq R_k}}
    {\bf S}_e
  \biggr)
\\&\qquad\qquad {} \;=\;
  \prod_{\substack{h\\ C_h \neq \emptyset}}
  \E \biggl(
  \bigwedge_{\substack{k \in C_h}}
  \bigvee_{\substack{e=(h,B) \in A_h \\ B\subseteq R_k}}
    {\bf S}_e
  \biggr)
\end{aligned}
\label{ineq:keep-trivial-loops}
\end{align}
The expression above has the form $\prod_h \E\Psi_h$.
We rewrite $\Psi_h$, by essentially enumerating all of its models
and checking if they satisfy~$\Psi_h$.
The result is the following equivalent form:
\begin{align*}
  \bigvee_{\substack{
    E_1\\
    E_1\subseteq A_h\\
    \forall k\in C_h,\; E_1\cap D_k\ne\emptyset
  }}
  \Biggl(
    \biggl(
      \bigwedge_{e\in E_1} {\bf S}_e
    \biggr)
    \land
    \biggl(
      \bigwedge_{e\in A_h \setminus E_1} \bar{\bf S}_e
    \biggr)
  \Biggr)
\end{align*}
%\HY{I cannot figure out why the above formula is the same as $\Psi_h$.}
and so
\begin{align}
\E \Psi_h
\;=\;
  \hskip-1.5em
  \sum_{\substack{
    E_1\\
    E_1\subseteq A_h\\ 
    \forall k\in C_h,\; E_1\cap D_k\ne\emptyset
  }}
  \hskip-1em
      \prod_{e\in E_1} \E {\bf S}_e
      \prod_{e\in A_h \setminus E_1} \E \bar{\bf S}_e
\label{eq:one-vertex-justified-prob}
\end{align}
Finally, we multiply the inequality \eqref{ineq:keep-trivial-loops}
  on both sides by $\prod_{e\in N} \E\bar{\bf S}_e$,
  plug in \eqref{eq:one-vertex-justified-prob},
  and use \autoref{lemma:likelihood}.
\end{proof}

Note that the upper bound is tight
  if $G^{\bot}$~has no cycles and therefore all loops are trivial.

\begin{proof}[Proof of the Lower Bound in \autoref{theorem:learning}]
By \autoref{lemma:likelihood}, 
\begin{align*}
\begin{aligned}
&\Pr\biggl( \bigcap_{k=1}^n \bigl( R_k = \reach{\randomH}{T_k} \bigr) \biggr)
  =
\\&\qquad
  \prod_{e\in N} \E \bar{\bf S}_e
  \cdot
  \E \biggl(
  \bigwedge_{\substack{L\\\text{\rm loop of~$G^{\bot}$}}}
  \bigwedge_{\substack{k\\ L\subseteq R_k\setminus T_k}}
  \bigvee_{\substack{e=(h,B)\\ e\in J_{G^{\bot}}(L)\setminus N\\ B\subseteq R_k}}
    {\bf S}_e
  \biggr)
\end{aligned}
\end{align*}
Thus, the main part of the lemma follows if we show that
\begin{align}
&
  \bigwedge_{\substack{h \\ C_h \neq \emptyset}}
  \bigvee_{\substack{
     E_1 
     \\ E_1 \subseteq A_h
     \\ \forall k \in C_h,\, E_1 \cap F_k \neq \emptyset}}
  \biggl(
    \Bigl(\bigwedge_{e \in E_1} {\bf S}_e\Bigr)
    \wedge
    \Bigl(\bigwedge_{e \in A_h \setminus E_1} \bar{\bf S}_e\Bigr)
  \biggr)
  \label{eq:likelihood-lbound}
\intertext{implies}
&
  \bigwedge_{\substack{L\\\text{\rm loop of~$G^{\bot}$}}}
  \bigwedge_{\substack{k\\ L\subseteq R_k\setminus T_k}}
  \bigvee_{\substack{e=(h,B)\\ e\in J_{G^{\bot}}(L)\setminus N\\ B\subseteq R_k}}
    {\bf S}_e
\label{eq:likelihood-lbound-exact}
\end{align}
To show this implication,
  we will show that a fixed but arbitrary conjunct of~\eqref{eq:likelihood-lbound-exact} holds,
  assuming that \eqref{eq:likelihood-lbound}~holds.
A conjunct of~\eqref{eq:likelihood-lbound-exact}
  is determined by a loop~$L_0$ and an index~$k_0$.
The idea is to show that
  loop $L_0$ is justified via its vertex that is closest to~$T_{k_0}$.

Since $L_0$~and~$k_0$ determine a conjunct of~\eqref{eq:likelihood-lbound-exact},
  we know that  $L_0 \subseteq R_{k_0} \setminus T_{k_0}$.
We need to find an arc $e = (h,B)$ such that
\begin{equation}
\label{eq:lboundlearning1a}
e \in J_{G^\bot}(L_0) \setminus N,
\quad
B \subseteq R_{k_0},
\quad\mbox{and}\quad
{\bf S}_e = 1.
\end{equation}
Since $L_0$ is not empty and $L_0 \subseteq R_{k_0} \subseteq \reach{G^\bot}{T_{k_0}}$,
we can choose $h \in L_0$ such that $d_{T_{k_0}}(h)$ is minimum.
Since $h \in L_0 \subseteq R_{k_0} \setminus T_{k_0}$, we have that
$$k_0 \in C_h.$$
This lets us instantiate~\eqref{eq:likelihood-lbound} with $h$,
and derive that for some subset $E_1$ of $A_h$,
\begin{equation}
\label{eq:lboundlearning2}
\text{$E_1 \cap F_k \neq \emptyset\,$ for all $k \in C_h$}
\quad\mbox{and}\quad \text{${\bf S}_e = 1\,$ for all $e \in E_1$} 
\end{equation}
Since $k_0 \in C_h$, the first conjunct implies that $E_1 \cap F_{k_0} \neq \emptyset$.
Thus, there exists an arc $e_0 = (h_0,B_0)$ in $E_1 \cap F_{k_0}$, and it satisfies
the following conditions:
\begin{enumerate}
\item the head $h_0$ of $e_0$ is $h$; 
\item $e_0$ is not in $N$; \RG{Why in $N$?}
\item $B_0 \subseteq R_{k_0}$; and
\item $e_0$ is a forward arc with respect to $T_{k_0}$.
\end{enumerate}
Since $e_0$ is a forward arc w.r.t. $T_{k_0}$ and $h$ has
the minimal distance from $T_{k_0}$ among all the vertices in $L_0$,
$$
e_0 \in J_{G^\bot}(L_0)
$$
Also, by the second conjunct in \eqref{eq:lboundlearning2},
$$
S_{e_0} = 1
$$
From what we have just shown follows that $e_0$ is the desired arc; it 
satisfies the requirements in \eqref{eq:lboundlearning1a}.
\end{proof}

Note that the lower bound and the upper bound coincide
  if $D_k \cap A_h = F_k \cap A_h$ for all $k$~and~$h$.
In this case, both bounds are tight.

}%\allowdisplaybreaks[1]

  \section{Proofs for Results in \autoref{sec:refinement}}
\label{section:proofs-refinement}

\noproblemma*
\begin{proof}{\allowdisplaybreaks[1]
The proof is a straightforward calculation.
\begin{align*}
  \Pr(a' \in {\bf A}_{\rm n}^{\approx})
&{}\;=\;
  \sum_{\substack{H'\\ H'\subseteq G^a}} [q \in \reach{H'}{(T(a,a'))}] \Pr(H')
\\
&{}\;\ge\;
  \sum_{\substack{H'\\ H\subseteq H' \subseteq G^a}}
    [q \in \reach{H'}{(T(a,a'))}] \Pr(H')
\\
&{}\;=\;
  \sum_{\substack{H'\\ H\subseteq H' \subseteq G^a}} \Pr(H')
  \;=\;
  \prod_{e\in H} \E {\bf S}_e
\end{align*}
The second equality uses two facts:
(i)~$q\in \reach{H}{(T(a,a'))}$,
and (ii)~$\reach{H}{(T(a,a'))} \subseteq \reach{H'}{(T(a,a'))}$ for all $H' \supseteq H$.
}\end{proof}

\greedylemma*
\begin{proof}
Another way of saying that the sequence $\{p_{\sigma(i)}/c_{\sigma(i)}\}_i$ is nonincreasing
  is to require that for all $1 \le i,j \le m$,
\begin{align}
i \le j \quad\Implies\quad
  \frac{c_{\sigma(i)}}{p_{\sigma(i)}} \le \frac{c_{\sigma(j)}}{p_{\sigma(j)}}
\label{eqn:greedy1}
\end{align}

Pick an arbitrary permutation $\sigma$.
We will study the effect of one transposition $(i\leftrightarrow i+1)$ on the cost.
Let $\sigma'=\sigma\circ(i\leftrightarrow i+1)$; in other words
\begin{align*}
  \sigma'(j) =
    \begin{cases}
    \sigma(i+1) &\text{if $j=i$} \\
    \sigma(i) &\text{if $j=i+1$} \\
    \sigma(j) &\text{otherwise}
    \end{cases}
\end{align*}
Observe that $q_k(\sigma)$ and $q_k(\sigma')$ differ for only \emph{one} value of~$k$:
\begin{align*}
  q_k(\sigma') &=
    \begin{cases}
    q_i(\sigma)(1-p_{\sigma(i+1)}) &\text{if $k=i+1$} \\
    q_k(\sigma) &\text{otherwise}
    \end{cases}
\end{align*}
Also notice that $q_k(\sigma) \neq 0$ and $q'_k(\sigma) \neq 0$ for all $k$.
The difference in cost between $\sigma'$ and $\sigma$ is
\begin{align*}
C(\sigma')-C(\sigma)
& {} =
\begin{aligned}[t]
& q_i(\sigma')c_{\sigma'(i)} + q_{i+1}(\sigma')c_{\sigma'(i+1)}
\\
& {} - q_i(\sigma)c_{\sigma(i)} - q_{i+1}(\sigma)c_{\sigma(i+1)}
\end{aligned}
\\[1ex]
& {} =
\begin{aligned}[t]
& q_i(\sigma)c_{\sigma(i+1)} + q_i(\sigma)(1 - p_{\sigma(i+1)})c_{\sigma(i)}
\\
&
{} - q_i(\sigma)c_{\sigma(i)} - q_i(\sigma)(1 - p_{\sigma(i)})c_{\sigma(i+1)}
\end{aligned}
\\[1ex]
& {} =
q_i(\sigma)(p_{\sigma(i)}c_{\sigma(i+1)} - p_{\sigma(i+1)}c_{\sigma(i)}).
\end{align*}
Thus,
$$
  \frac{C(\sigma')-C(\sigma)}{q_i(\sigma)} = p_i c_{i+1}-p_{i+1}c_i
$$
where $p_i$ denotes $p_{\sigma(i)}$, and $c_i$ denotes $c_{\sigma(i)}$,
for the fixed permutation~$\sigma$.

All that remains is to interpret the result of these calculations.
For the left-to-right direction,
assume that $\sigma$ has the minimal cost.
Also, for the sake of contradiction, suppose that
there exist $i$~and~$j$ such that
  $i \leq j$ and $c_i/p_i>c_j/p_j$.
Then, there must also exist an $i$ such that $c_i/p_i>c_{i+1}/p_{i+1}$,
which is equivalent to
$$
p_i c_{i+1} - p_{i+1} c_i < 0.
$$
Thus, the previous calculation shows that $\sigma'$ would have a lower
cost than $\sigma$.
This contradicts the assumption that $\sigma$ has the minimal cost.

For the right-to-left direction, pick $\sigma$ and $\sigma'$ that satisfy the RHS
of \eqref{eqn:greedy1}. Then, we can convert $\sigma$ to $\sigma'$ by composing
$\sigma$ with a sequence of transpositions $i \leftrightarrow i+1$
for $i$ such that
$$
\frac{c_i}{p_i} = \frac{c_{i+1}}{p_{i+1}}.
$$
Then the previous computation
shows that such composition leaves the cost unchanged.
Thus, $\sigma$ and $\sigma'$ have the same cost. But by what we have already shown,
there should be at least one $\sigma''$ that satisfies the RHS of
\eqref{eqn:greedy1} and have the minimal cost. This implies that
all of $\sigma$, $\sigma'$ and $\sigma''$ are optimal.
\end{proof}

\lemmafeasible*
\begin{proof}
The inclusion $F(G^a_{\to}) \subseteq F(G^a)$ follows 
from $G^a_{\to} \subseteq G^a$. We have $(\top,G^a_{\to}) \in F(G^a_{\to})$ because
(a)~$\top>a$ by assumption,
(b)~$G^a_{\to}\subseteq G^a_{\to}$ trivially, and
(c)~$q \in \reach{G^a_{\to}}{(T(a,\top))}$. To see why (c) holds,
notice that
  removing nonforward arcs with respect to $T(a,\top)=P_0(a) \cup P_1(a)$
  preserves distances and reachability from $T(a,\top)$,
  and so $\reach{G^a_{\to}}{(T(a,\top))} = \reach{G^a}{(T(a,\top))}$.
\end{proof}

\bijectionlemma*
%\RG{Maybe fill in the details of the proof if there is time.}
\begin{proof}[Proof sketch]
Let
\begin{align*}
G' &{}\;\defeq\; \{\,(h,{e}\cup B) \mid e = (h,B) \in G^a_{\to}\,\}
\\
S' &{}\;\defeq\;\{\,e \mid e \in G^a_{\to}\,\}
\end{align*}
Because $G^a_{\to}$ has no cycles by construction,
$G'$ does not have cycles, either. We have that 
$$
\Bigl(\mathop{\exists}_{e\in G^a_{\to}}\!\! y_e \; (\Phi_1 \land \Phi_2)\Bigr)
\;\;\Iff\;\;
\Bigl(\phi^{S' \cup P_0(a)\cup P_1(a)}(G')\Bigr)
$$
where the latter uses the definition in~\eqref{eq:graph-formula-relative}.
Thus, we can apply \autoref{corollary:loop-formulas-flexible}.
Finally, note that $\Phi_3$ ensures that $a'>a$ and $q\in\reach{H}{(T(a,a'))}$.
\end{proof}

% vim:spell:

}{}
\end{document}